\UseRawInputEncoding
\documentclass[12pt]{amsart}
\usepackage{}

\usepackage{amsmath}
\usepackage{amsfonts}
\usepackage{amssymb}
\usepackage[all]{xy}           

\usepackage{bbding}
\usepackage{txfonts}
\usepackage{amscd}
\usepackage{footnote}

\usepackage[shortlabels]{enumitem}
\usepackage{ifpdf}
\ifpdf
  \usepackage[colorlinks,final,backref=page,hyperindex]{hyperref}
\else
  \usepackage[colorlinks,final,backref=page,hyperindex,hypertex]{hyperref}
\fi
\usepackage{tikz}
\usepackage[active]{srcltx}

\makeatletter

\newtheorem{thm}{Theorem}[section]

\newtheorem{pro}[thm]{Proposition}
\newtheorem{ex}[thm]{Example}
\newtheorem{rmk}[thm]{Remark}
\newtheorem{defi}[thm]{Definition}

\newcommand {\emptycomment}[1]{}

\newcommand{\lon }{\,\rightarrow\,}
\newcommand{\be }{\begin{equation}}
\newcommand{\ee }{\end{equation}}

\newcommand{\g}{\mathfrak g}
\newcommand{\h}{\mathfrak h}

\newcommand{\huaB}{\mathcal{B}}

\newcommand{\huaC}{{\mathcal{C}}}
\newcommand{\huaD}{\mathcal{D}}

\newcommand{\frks}{\mathfrak s}
\newcommand{\frkt}{\mathfrak t}

\newcommand{\dM}{\mathrm{d}}

\newcommand{\Id}{{\rm{Id}}}

\newcommand{\br}[1]{   [ \cdot,    \cdot  ]   }

\newcommand{\Der}{\mathrm{Der}}

\newcommand{\Ad}{\mathrm{Ad}}

\newcommand{\Aut}{\mathrm{Aut}}

\newcommand{\gl}{\mathfrak {gl}}

\newcommand{\ad}{\mathrm{ad}}

\begin{document}

\title[Rota-Baxter operators on crossed modules of Lie groups]{Rota-Baxter operators on crossed modules of Lie groups and categorical solutions of the Yang-Baxter equation}

\author{Jun Jiang}
\address{Department of Mathematics, Jilin University, Changchun 130012, Jilin, China}
\email{junjiang@jlu.edu.cn}


\begin{abstract}
In this paper, we construct a categorical solution $(\huaC, R)$ of the Yang-Baxter equation, i.e. $\huaC$ is a small category and $R: \huaC\times\huaC\lon\huaC\times\huaC$ is an invertible functor satisfying
$$
(R\times\Id_\huaC)(\Id_\huaC\times R)(R\times\Id_\huaC)=(\Id_\huaC\times R)(R\times\Id_\huaC)(\Id_\huaC\times R),
 $$
where $\huaC\times\huaC$ is the product category. First, the notion of Rota-Baxter operators on crossed modules of Lie groups is defined and its various properties are established. Then, we use Rota-Baxter operators on crossed modules of Lie groups to construct categorical solutions of the Yang-Baxter equation. We also study the Rota-Baxter operators on crossed modules of Lie algebras which are infinitesimals of Rota-Baxter operators on crossed modules of Lie groups, they can give connections on manifolds. Finally, we study the integration of Rota-Baxter operators on crossed modules of Lie algebras and the differentials of Rota-Baxter operators on crossed modules of Lie groups.
\end{abstract}

\renewcommand{\thefootnote}{}


\maketitle



\allowdisplaybreaks


\section{Introduction}

\subsection{The Yang-Baxter equation}
The Yang-Baxter equation was firstly introduced in theoretical physics by Yang\cite{Y} in 1967. Shortly afterwards, this equation appeared in statistical mechanics in the work of Baxter\cite{BAX}. The Yang-Baxter equation has led to many interesting applications in mathematics, such as quantum groups, tensor categories, knot theory and integrable systems. Recall that a solution of the Yang-Baxter equation on a vector space $V$ is an invertible linear map $R:V\otimes V\lon V\otimes V$ such that
$$
(R\otimes\Id_V)(\Id_V\otimes R)(R\otimes\Id_V)=(\Id_V\otimes R)(R\otimes\Id_V)(\Id_V\otimes R).
$$
The problem of constructing and classifying its solutions is difficult, so Drinfel'd raised the question of finding set-theoretical solutions of the Yang-Baxter equation\cite{Dr}. A set-theoretical solution of the Yang-Baxter equation is a pair $(X, r)$, where X is a set and $r:X\times X\lon X\times X$ is a bijective map satisfying
$$
(r\times\Id_X)(\Id_X\times r)(r\times\Id_X)=(\Id_X\times r)(r\times\Id_X)(\Id_X\times r).
$$
There are many results on set-theoretical solutions of the Yang-Baxter equation by Etingof, Schedler and Soloviev\cite{ESS}, Lu, Yan and Zhu\cite{LYZ}, Gateva-Ivanova and Majid\cite{GM}, Gateva-Ivanova and Cameron\cite{GaC}, Gateva-Ivanova\cite{Ga} and others. In \cite{Ru}, Rump introduced the notion of braces, a generalization of Jacobson radical rings as a tool to construct non-degenerate involutive set-theoretic solutions of the Yang-Baxter equation. Guarnieri and Vendramin generalized braces to skew braces and used them to construct non-degenerate non-involutive set-theoretic solutions of the Yang-Baxter equation\cite{GV}. Recently, Bardakov and Gubarev studied the connection between skew braces and Rota-Baxter operators on Lie groups, they obtained non-degenerate set-theoretic solutions of the Yang-Baxter equation\cite{BG2}.
\subsection{Crossed modules of Lie groups and crossed modules of Lie algebras}
The notion of crossed modules of groups was introduced by Whitehead in order to characterize a second relative homotopy group\cite{Whi, WHi}. Crossed modules of groups and cohomologies of groups are closely related, they can be classified by the third cohomology group of groups. When people replace groups and maps by Lie groups and smooth maps in the definition of crossed modules of groups, they obtain crossed modules of Lie groups. The infinitesimal of crossed modules of Lie groups are crossed modules of Lie algebras, which can be classified by the third cohomology group of Lie algebras. Furthermore, every crossed module of Lie algebras can be integrated into a crossed module of Lie groups. For further details on crossed modules of algebraic structures, see \cite{WF}. When studying the categorification, Baez and his collaborators finded that a strict Lie $2$-group(resp. strict Lie $2$-algebra) is ono-to-one correspondence to a crossed module of Lie groups (resp. crossed module of Lie algebras)\cite{Baez, BaL}, this means that there are small categories associated with crossed modules of Lie groups. These small categories play an important role in this paper. Recently, crossed module of Lie groups have become an object of intense study in the higher gauge theory, $2$-BF theory\cite{BaezSch, Gas, MM}.

\subsection {Rota-Baxter operators}The notion of Rota-Baxter operators on associative algebras was introduced by G. Baxter in his study of fluctuation theory in probability \cite{Bax}. It was further developed by Rota and Cartier from the perspective of combinatorics \cite{Ra1, Ra2, Car}. Recently it has found many applications, including Connes-Kreimer's algebraic approach to renormalization of quantum field theory~\cite{CK}. In the Lie algebra context, the notion of Rota-Baxter operators on Lie algebras was introduced in \cite{Ku,BGN}.
 Let $(\g, [\cdot, \cdot]_\g)$ be a Lie algebra and $\lambda\in\mathbb{R}$ be a fixed element. A linear map $B:\g\lon\g$ is called a Rota-Baxter operator of weight $\lambda$ if
 $$
 [B(x), B(y)]_\g=B([B(x), y]_\g+[x, B(y)]_\g+\lambda[x, y]_\g), \quad \forall x, y\in\g.
 $$
Rota-Baxter operators on Lie algebras are closely related to various classical Yang-Baxter equations. More precisely, if $(\g, [\cdot, \cdot]_\g)$ has a non-degenerate symmetric invariant bilinear form, a Rota-Baxter operator of weight $0$ on $\g$ is naturally the operator form of a skew-symmetric solution of the classical Yang-Baxter equation\cite{STS}. Rota-Baxter operators of weight $1$ are in one-to-one
correspondence with solutions of the modified Yang-Baxter equation, and have close relation with factorizable Lie bialgebras \cite{LS}. In addition, a Lie group with zero Schouten curvature gives rise a Rota-Baxter operator of weight $1$ on a Lie algebra \cite{Kos}.

The notion of Rota-Baxter operators on  Lie groups  was introduced in \cite{GLS}, then this concept was generalized to relative Rota-Baxter operators on Lie groups in \cite{JSZ}.
Let $(G, \cdot)$ be a Lie group. A smooth map $\huaB: G\lon G$ is called Rota-Baxter operator on $G$ if
\begin{equation*}
\huaB(a)\cdot\huaB(b)=\huaB(a\cdot\huaB(a)\cdot b\cdot\huaB(a)^{-1}), \quad\forall a, b\in G.
\end{equation*}
Similarly to Rota-Baxter operators on Lie algebras give solutions of the classical Yang-Baxter equation, Rota-Baxter operators on Lie groups can also be used to construct set-theoretical solutions of the Yang-Baxter equation \cite{BG2}. Rota-Baxter operators on Lie groups also give rise to factorizations of the Lie groups, skew braces, and can be applied to study integrable systems \cite{GLS, BG2, BNMY, CS, NM, STS}.

According to Lie's third theorem, there is a one-to-one correspondence between finite dimensional Lie algebras and connected and simply connected Lie groups. However, this theorem does not always hold for Rota-Baxter operators. In deed, it was shown in \cite{GLS,JSZ} that by differentiating Rota-Baxter operators on Lie groups, people can obtain Rota-Baxter operators on the corresponding Lie algebras. The converse direction, Rota-Baxter operators on Lie algebras can be integrated to local Rota-Baxter operators on Lie groups, global integrability of Rota-Baxter operators on Lie algebras have obstructions\cite{JSZ, JSZ1}.

\subsection{Outline of the paper}
A categorical solution of the Yang-Baxter equation is a pair $(\huaC, R)$, where $\huaC$ is a small category and $R: \huaC\times\huaC\lon\huaC\times\huaC$ is an invertible functor satisfying
$$
(R\times\Id_\huaC)(\Id_\huaC\times R)(R\times\Id_\huaC)=(\Id_\huaC\times R)(R\times\Id_\huaC)(\Id_\huaC\times R).
 $$
Inspired by that Rota-Baxter operators on Lie groups can be used to construct set-theoretic solutions of the Yang-Baxter equation, a purpose of this paper is to consider using Rota-Baxter operators on crossed modules of Lie groups to construct categorical solutions $(\huaC, R)$ of the Yang-Baxter equations, where $\huaC$ is a small category associated with a crossed module of Lie groups.

In section 2, we recall the basic notions and facts about crossed modules of Lie groups and introduce Rota-Baxter operators on crossed module of Lie groups. Rota-Baxter operators on crossed modules of Lie groups are also characterized by graphs(Theorem \ref{graphB}). On the one hand, Rota-Baxter operators on crossed modules of Lie groups descend crossed modules of Lie groups(Proposition \ref{dg}), these descendent crossed modules of Lie groups can induce semi-direct product Lie groups. On the other hand, Rota-Baxter operators on crossed modules of Lie groups induce Rota-Baxter operators on semi-direct product Lie groups(Theorem \ref{semirb}), these Rota-Baxter operators on semi-direct product Lie groups can descend Lie groups. The Lie groups given by the above two approaches are isomorphic(Proposition \ref{isomdi}). We summarize the constructions and relations in the following diagram.
 
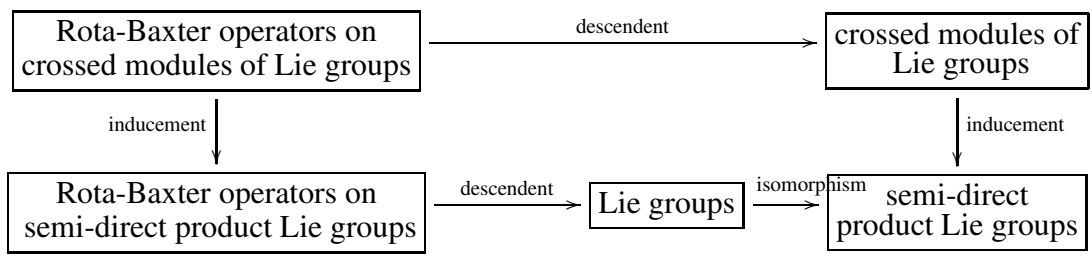
\begin{equation*}
\xymatrix{
  \boxed{\txt{\text{Rota-Baxter operators on} \\\text{crossed modules of Lie groups}}} \ar[d]_{\text{inducement}} \ar[rrr]^-{\text{descendent}} & & & \boxed{\txt{\text{crossed modules of}\\\text{ Lie groups}}} \ar[d]^{\text{inducement}} \\
  \boxed{\txt{\text{Rota-Baxter operators on}\\\text{ semi-direct product Lie groups}}} \ar[rr]^-{\text{descendent}}  & & \boxed{\text{Lie groups}} \ar[r]^{\text{isomorphism}} &  \boxed{\txt{\text{semi-direct}\\\text{product Lie groups}}}
}
\end{equation*}

In section 3, we introduce the notion of categorical solutions of the Yang-Baxter equation, then categorical solutions of the Yang-Baxter equation are constructed by Rota-Baxter operators on crossed modules of Lie groups(Theorem \ref{categorical}).

In section 4, first, we recall the basic notions and facts about crossed modules of Lie algebras, then we introduce Rota-Baxter operators on crossed modules of Lie algebras. We establish some properties of Rota-Baxter operators on crossed modules of Lie algebras parallel to Rota-Baxter operators on crossed modules of Lie groups. These properties can be summarized by the following diagram.
 
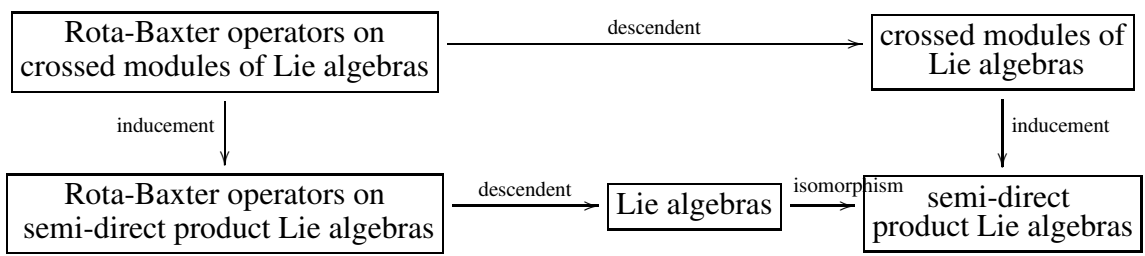
\begin{equation*}
\xymatrix{
  \boxed{\txt{\text{Rota-Baxter operators on} \\\text{crossed modules of Lie algebras}}} \ar[d]_{\text{inducement}} \ar[rrr]^-{\text{descendent}} & & & \boxed{\txt{\text{crossed modules of}\\\text{ Lie algebras}}} \ar[d]^{\text{inducement}} \\
  \boxed{\txt{\text{Rota-Baxter operators on}\\\text{ semi-direct product Lie algebras}}} \ar[rr]^-{\text{descendent}}  & & \boxed{\text{Lie algebras}} \ar[r]^{\text{isomorphism}} &  \boxed{\txt{\text{semi-direct }\\\text{product Lie algebras}}}
}
\end{equation*}
Finally, connections on manifolds are constructed by Rota-Baxter operators on crossed modules of Lie algebras(Proposition \ref{2-connection}).

In section 5, we study the integration of Rota-Baxter operators on crossed modules of Lie algebras(Theorem \ref{i1}) and the differentials of Rota-Baxter operators on crossed modules of Lie groups(Proposition \ref{thmcrlg}, Theorem \ref{d1}, Proposition \ref{d2}).

Throughout the paper, all the vector spaces are over $\mathbb{R}$ and finite-dimensional, all manifolds are smooth manifolds.

\section{Rota-Baxter operators on crossed modules of Lie groups}
In this section, first we recall Rota-Baxter operators on Lie groups and crossed modules of Lie groups. Then the notion of Rota-Baxter operators on crossed modules of Lie groups is defined and its various properties are established.
\begin{defi}\rm(\cite{GLS})
Let $(G, \cdot)$ be a Lie group. A smooth map $\huaB: G\lon G$ is called Rota-Baxter operator on $G$ if
\begin{equation*}
\huaB(a)\cdot\huaB(b)=\huaB(a\cdot\huaB(a)\cdot b\cdot\huaB(a)^{-1}), \quad\forall a, b\in G.
\end{equation*}
\end{defi}

\begin{ex}
Let $(G, \cdot)$ be a Lie group. Then the inverse map $\huaB=(\cdot)^{-1}$ is a Rota-Baxter operator on the Lie group $(G, \cdot)$.
\end{ex}

\begin{defi}\rm(\cite{GLS})
Let $\huaB, \huaB'$ be Rota-Baxter operators on Lie groups $G, H$ respectively. A smooth map $\Psi:H\lon G$ is called Rota-Baxter operator homomorphism from $\huaB'$ to $\huaB$ if $\Psi$ is a Lie group homomorphism and
$$\Psi\circ \huaB' =\huaB\circ\Psi.$$
\end{defi}

\begin{defi}\rm(\cite{Whi, WHi})
A crossed module of Lie groups is a quadruple $(H, G, t, \Phi)$, where $(H, \cdot_H)$ and $(G, \cdot_G)$ are Lie groups, $t:H\lon G$ is a Lie group homomorphism and $\Phi:G\lon\Aut(H)$ is an action of $(G, \cdot_G)$ on $(H, \cdot_H)$, such that
\begin{eqnarray}
\label{crmo1}\Phi(t(p))q&=&p\cdot_Hq\cdot_H p^{-1}, \\
\label{crmo2}t(\Phi(a)p)&=&a\cdot_G t(p)\cdot_G a^{-1}, \quad \forall p, q\in H, a\in G.
\end{eqnarray}
\end{defi}

In \cite{NK}, Katherine Norrie studied the actions and semi-direct products of crossed modules of Lie groups, then he obtained the following proposition.

\begin{pro}\rm(\cite{NK, NK1})\label{croslie}
Let $(H, G, t, \Phi)$ be a crossed module of Lie groups. Denote by $(G\ltimes G, \cdot_{G\ltimes})$ and $(H\ltimes H, \cdot_{H\ltimes})$ the semi-direct product Lie groups, where $G\ltimes G=G\times G, H\ltimes H=H\times H$ and
\begin{eqnarray}
\label{defigs}(a, b)\cdot_{G\ltimes}(c, d)&=&(a\cdot_G c, b\cdot_G a\cdot_G d\cdot_G a^{-1}), \quad a, b, c, d\in G,\\
\label{defigs1}(p, q)\cdot_{H\ltimes}(r, s)&=&(p\cdot_H r, q\cdot_H p\cdot_H s\cdot_H p^{-1}), \quad p, q, r, s\in H.
\end{eqnarray}
Then $(H\ltimes H, G\ltimes G, (t, t), \tilde{\Phi})$ is a crossed module of Lie groups, where
$$
(t, t):H\ltimes H\lon G\ltimes G, \quad (t, t)(p, q)=(t(p), t(q)),
$$
and $\tilde{\Phi}:G\ltimes G\lon\Aut(H\ltimes H)$ is defined by
\begin{equation}\label{eqdefrG}
\tilde{\Phi}\Big((a, b)\Big)(p, q)=\Big(\Phi(a)p, \Phi(b\cdot_G a)(q\cdot_H p)\cdot_H\Phi(a)p^{-1}\Big), \quad \forall a, b\in G, p, q\in H.
\end{equation}
\end{pro}

\begin{defi}
Let $(H, G, t, \Phi)$ be a crossed module of Lie groups. A Rota-Baxter operator on $(H, G, t, \Phi)$ is a pair $(\huaB_1, \huaB_0)$, where smooth maps $\huaB_1: H\lon H$ and $\huaB_0:G\lon G$ satisfy
\begin{itemize}
\item [\rm(i)] $\huaB_1$ and $\huaB_0$ are Rota-Baxter operators on $H$ and $G$ respectively,
\item [\rm(ii)]the smooth map $t$ is a homomorphism from $\huaB_1$ to $\huaB_0$, i.e. $t\circ \huaB_1=\huaB_0\circ t$,
\item [\rm(iii)] $\Phi(\huaB_0(a))\huaB_1(p)=\huaB_1\Big(\Phi(a\cdot_G\huaB_0(a))(p\cdot_H\huaB_1(p))\cdot_H\Phi(\huaB_0(a))\huaB_1(p)^{-1}\Big),
    \quad \forall a\in G, p\in H.$
\end{itemize}
\end{defi}

\begin{ex}
Let $(H, G, t, \Phi)$ be a crossed module of Lie groups. Then $\Big((\cdot)_H^{-1}, (\cdot)_G^{-1}\Big)$ is a Rota-Baxter operator on $(H, G, t, \Phi)$, where $(\cdot)_H^{-1}$ and $(\cdot)_G^{-1}$ are inverse maps of $H$ and $G$ respectively.
\end{ex}

\begin{ex}\label{ex111}
Let $H$ be a Lie group and $\huaB$ be a Rota-Baxter operator on $H$. Then $(H, \{e_H\}, t, \Phi)$ is a crossed module of Lie groups, where
$$\Phi(e_H)h=h, ~~~\forall h\in H, \quad t(H)=e_H.$$
Moreover, $(\huaB, \Id)$ is a Rota-Baxter operator on $(H, \{e_H\}, t, \Phi)$.
\end{ex}

\begin{ex}
Let $G$ be a Lie group. Define $t:G\lon G$ and $\Phi: G\lon\Aut(G)$ by
$$
t(a)=a, \quad \Phi(a)b=aba^{-1}, \quad \forall a, b\in G.
$$
Then $(G, G, t, \Phi)$ is a crossed module of Lie groups. If $\huaB: G\lon G$ is a Rota-Baxter operator on $G$, then $(\huaB, \huaB)$ is a Rota-Baxter operator on $(G, G, t, \Phi)$.
\end{ex}

Let $G$ be a Lie group and $\g$ be the Lie algebra of $G$. Define $t:\g\lon G$ and $\Ad: G\lon\mathrm{GL}(\g)$ by
$$t(x)=e, ~~\forall x\in\g, \quad \Ad(g)x=\frac{d}{dt}|_{t=0}g\exp(tx)g^{-1}, ~~\forall g\in G.$$
Then $(\g, G, t, \Ad)$ is a crossed module of Lie groups. If $\huaB$ is a Rota-Baxter operator on $G$, we have \begin{eqnarray*}
\huaB(a)\huaB(b)\huaB(a)^{-1}&=&\huaB(a)\huaB(b)\huaB(\huaB(a)^{-1}a^{-1}\huaB(a))\\
&=&\huaB\Big(a\huaB(a)b\huaB(b)\huaB(a)^{-1}a^{-1}\huaB(a)\huaB(b)^{-1}\huaB(a)^{-1}\Big).
\end{eqnarray*}
Then for any $x\in\g$, it follows
\begin{eqnarray*}
\Ad_{\huaB(a)}\huaB_{*e}(x)&=&\frac{d}{dt}|_{t=0}\huaB(a)\huaB(\exp(tx)))\huaB(a)^{-1}\\
&=&\huaB_{*e}\Big(\Ad_{a\huaB(a)}(x+\huaB_{*e}(x))-\Ad_{\huaB(a)}(\huaB_{*e}(x))\Big).
\end{eqnarray*}
Thus we have the following example.
\begin{ex}
If $\huaB$ is a Rota-Baxter operator on $G$, then $((\huaB)_{*e}, \huaB)$ is a Rota-Baxter operator on $(\g, G, t, \Ad)$.
\end{ex}

In the following, we will character Rota-Baxter operators on crossed modules of Lie groups as crossed module of Lie groups.

\begin{thm}\label{graphB}
Let $(H, G, t, \Phi)$ be a crossed module of Lie groups and $\huaB_1: H\lon H, \huaB_0: G\lon G$ be smooth maps. Then the pair $(\huaB_1, \huaB_0)$ is a Rota-Baxter operator on $(H, G, t, \Phi)$ if and only if $(\mathrm{Gr}(\huaB_1), \mathrm{Gr}(\huaB_0), (t, t), \tilde{\Phi})$ is a crossed module of Lie groups, where
$$
\mathrm{Gr}(\huaB_1)=\{(\huaB_1(p), p)|\forall p\in H\}\subset H\ltimes H, \quad \mathrm{Gr}(\huaB_0)=\{(\huaB_0(a), a)|\forall a\in G\}\subset G\ltimes G.
$$
$$
(t, t): \mathrm{Gr}(\huaB_1)\lon\mathrm{Gr}(\huaB_0), \quad (t, t)(\huaB_1(p), p)=(t(\huaB_1(p)), t(p))
$$
and $\tilde{\Phi}:\mathrm{Gr}(\huaB_0)\lon\Aut(\mathrm{Gr}(\huaB_1))$ defined by
$$
\tilde{\Phi}\Big((\huaB_0(a), a)\Big)(\huaB_1(p), p)=\Big(\Phi(\huaB_0(a))\huaB_1(p), \Phi(a\cdot_G \huaB_0(a))(p\cdot_H \huaB_1(p))\cdot_H\Phi(\huaB_0(a))\huaB_1(p)^{-1}\Big)
$$
\end{thm}
\begin{proof}
Suppose that $(\huaB_1, \huaB_0)$ is a Rota-Baxter operator on crossed module of Lie groups, it follows from \cite[Proposition 3.4]{JSZ} that $\mathrm{Gr}(\huaB_1)$ and $\mathrm{Gr}(\huaB_0)$ are Lie subgroups of $(H\ltimes H, \cdot_{H\ltimes})$ and $(G\ltimes G, \cdot_{G\ltimes})$.

Since $(t, t)$ is a Lie group homomorphism from $(H\ltimes H, \cdot_{H\ltimes})$ to $(G\ltimes G, \cdot_{G\ltimes})$, we have
$$
(t, t)(\huaB_1(p), p)=(t(\huaB_1(p)), t(p))=(\huaB_0(t(p)), t(p))\in\mathrm{Gr}(\huaB_0), \quad \forall p\in H,
$$
which implies that $(t, t)$ is a Lie group homomorphism from $\mathrm{Gr}(\huaB_1)$ to $\mathrm{Gr}(\huaB_0)$.

Moreover, since $(\huaB_1, \huaB_0)$ is a Rota-Baxter operator on $(H, G, t, \Phi)$, it follows that
\begin{eqnarray*}
&&\tilde{\Phi}\Big((\huaB_0(a), a)\Big)(\huaB_1(p), p)\\
&=&\Big(\Phi(\huaB_0(a))\huaB_1(p), \Phi(a\cdot_G\huaB_0(a))(p\cdot_H\huaB_1(p))\cdot_H\Phi(\huaB_0(a))\huaB_1(p)^{-1}\Big)\\
&=&\Big(\huaB_1\Big(\Phi(a\cdot_G\huaB_0(a))(p\cdot_H\huaB_1(p))\cdot_H\Phi(\huaB_0(a))\huaB_1(p)^{-1}\Big), \\
&&\Phi(a\cdot_G\huaB_0(a))(p\cdot_H\huaB_1(p))\cdot_H\Phi(\huaB_0(a))\huaB_1(p)^{-1}\Big)\in \mathrm{Gr}(\huaB_1),
\end{eqnarray*}
thus $\tilde{\Phi}$ is an action of $\mathrm{Gr}(\huaB_0)$ on $\mathrm{Gr}(\huaB_1)$ by Proposition \ref{croslie}.

As $(H, G, t, \Phi)$ is a crossed module of Lie groups, we have
\begin{eqnarray*}
&&\tilde{\Phi}\Big((t, t)(\huaB_1(p), p)\Big)(\huaB_1(q), q)\\
&=&\tilde{\Phi}\Big(t(\huaB_1(p)), t(p)\Big)(\huaB_1(q), q)\\
&=&\Big(\Phi(t(\huaB_1(p))\huaB_1(q), \Phi(t(p)\cdot_H t(\huaB_1(p)))(q\cdot_G\huaB_1(q))\cdot_H\Phi(t(\huaB_1(p)))\huaB_1(q)^{-1}\Big)\\
&=&\Big(\huaB_1(p)\cdot_H\huaB_1(q)\cdot_H\huaB_1(p)^{-1},\\
 &&p\cdot_H\huaB_1(p)\cdot_H q\cdot_H\huaB_1(q)\cdot_H\huaB_1(p)^{-1}\cdot_H p^{-1}\cdot_H\huaB_1(p)\cdot_H\huaB_1(q)^{-1}\cdot_H\huaB_1(p)^{-1}\Big)\\
&=&(\huaB_1(p), p)\cdot_{H\ltimes}(\huaB_1(q), q)\cdot_{H\ltimes}(\huaB_1(p), p)^{-1},
\end{eqnarray*}
and
\begin{eqnarray*}
&&(t, t)(\tilde{\Phi}(\huaB_0(a), a)(\huaB_1(p), p))\\
&=&(t, t)(\Phi(\huaB_0(a))\huaB_1(p), \Phi(a\cdot_G\huaB_0(a))(p\cdot_H\huaB_1(p))\cdot\Phi(\huaB_0(a))\huaB_1(p)^{-1})\\
&=&(\huaB_0(a)\cdot_G t(\huaB_1(p))\cdot_G\huaB_0(a)^{-1},\\
&&a\cdot_G\huaB_0(a)\cdot_G t(p\cdot_H\huaB_1(p))\cdot_G\huaB_0(a)^{-1}\cdot_G a^{-1}\cdot \huaB_0(a)\cdot_G t(\huaB_1(p)^{-1})\cdot_G\huaB_0(a)^{-1})\\
&=&(\huaB_0(a), a)\cdot_G(t, t)(\huaB_1(p), p)\cdot_G(\huaB_0(a), a)^{-1}.
\end{eqnarray*}
Thus $(\mathrm{Gr}(\huaB_1), \mathrm{Gr}(\huaB_0), (t, t), \tilde{\Phi})$ is a crossed module of Lie groups.

Suppose that $(\mathrm{Gr}(\huaB_1), \mathrm{Gr}(\huaB_0), (t, t), \tilde{\Phi})$ is a crossed module of Lie groups, by \cite[Proposition 3.4]{JSZ}, we obtain that $\huaB_1$ and $\huaB_0$ are Rota-Baxter operators on $H$ and $G$.

Since $(t, t):\mathrm{Gr}(\huaB_1)\lon\mathrm{Gr}(\huaB_0)$ is a Lie group homomorphism, it follows that $t\circ\huaB_1=\huaB_0\circ t$. Moreover, since $\tilde{\Phi}(\mathrm{Gr}(\huaB_0))\mathrm{Gr}(\huaB_1)\subseteq\mathrm{Gr}(\huaB_1)$, then
$$
\tilde{\Phi}(\huaB_0(a), a)(\huaB_1(p), p)=\Big(\Phi(\huaB_0(a))\huaB_1(p), \Phi(a\cdot_G\huaB_0(a))(p\cdot_H\huaB_1(p))\cdot_H\Phi(\huaB_0(a))\huaB_1(p)^{-1}\Big)\in\mathrm{Gr}(\huaB_1),
$$
which implies that
$$
\Phi(\huaB_0(a))\huaB_1(p)=\huaB_1\Big(\Phi(a\cdot_G\huaB_0(a))(p\cdot_H\huaB_1(p))\cdot_H\Phi(\huaB_0(a))\huaB_1(p)^{-1}\Big).
$$
Therefore, $(\huaB_1, \huaB_0)$ is a Rota-Baxter operator on $(H, G, t, \Phi)$.
\end{proof}

Similar to a Rota-Baxter operator $\huaB$ on a Lie group $G$ descends a new Lie group structure on $G$, a Rota-Baxter operator on a crossed module of Lie groups also descend a new crossed module of Lie groups.
\begin{pro}\label{dg}
Let $(\huaB_1, \huaB_0)$ be a Rota-Baxter operator on a crossed module of Lie groups $(H, G, t, \Phi)$. Then $\Big((H, \cdot_{\huaB_1}), (G, \cdot_{\huaB_0}), t, \hat{\Phi}\Big)$ is a crossed module of Lie groups, where
\begin{eqnarray}
\label{deslie1}p\cdot_{\huaB_1}q&=&p\cdot_H\huaB_1(p)\cdot_H q\cdot_H\huaB_1(p)^{-1}, \quad \forall p, q\in H,\\
\label{deslie2}a\cdot_{\huaB_0}b&=&a\cdot_G\huaB_0(a)\cdot_G b\cdot_G\huaB_0(a)^{-1}, \quad \forall a, b\in G,\\
\label{deslie3}\hat{\Phi}(a)p&=&\Phi(a\cdot_G\huaB_0(a))(p\cdot_H\huaB_1(p))\cdot_H\Phi(\huaB_0(a))\huaB_1(p)^{-1}, \quad \forall a\in G, p\in H.
\end{eqnarray}
\end{pro}
\begin{proof}
By \cite[Proposition 3.5]{JSZ}, we have that $(H, \cdot_{\huaB_1}), (G, \cdot_{\huaB_0})$ are Lie groups. Since $t$ is a Lie group homomorphism and $t\circ\huaB_1=\huaB_0\circ t$, it follows that $t$ is a Lie group homomorphism from $(H, \cdot_{\huaB_1})$ to $(G, \cdot_{\huaB_0})$.

On the one hand, since $(\huaB_1, \huaB_0)$ is a Rota-Baxter operator on $(H, G, t, \Phi)$, we have
\begin{eqnarray*}
&&\hat{\Phi}(a)(p\cdot_{\huaB_1}q)=\hat{\Phi}(a)(p\cdot_H\huaB_1(p)\cdot_H q\cdot_H\huaB_1(p)^{-1})\\
&=&\Phi(a\cdot_G\huaB_0(a))(p\cdot_H\huaB_1(p)\cdot_H q\cdot_H\huaB_1(p)^{-1}\cdot_H\huaB_1(p)\cdot_H\huaB_1(q))\Phi(\huaB_0(a))(\huaB_{1}(q)^{-1}\cdot_H\huaB_{1}(p)^{-1})\\
&=&\Phi(a\cdot_G\huaB_0(a))(p\cdot_H\huaB_1(p))\cdot_H\Phi(a\cdot_G\huaB_0(a))(q\cdot_H\huaB_1(q))\\
&&\cdot_H\Phi(\huaB_0(a))\huaB_{1}(q)^{-1}\cdot_H\Phi(\huaB_0(a))\huaB_{1}(p)^{-1},
\end{eqnarray*}
and
\begin{eqnarray*}
&&\hat{\Phi}(a)(p)\cdot_{\huaB_1}\hat{\Phi}(a)(q)\\
&=&\hat{\Phi}(a)(p)\cdot_{H}\huaB_1(\hat{\Phi}(a)(p))\cdot_{H}\hat{\Phi}(a)(q)\cdot_H\huaB_1(\hat{\Phi}(a)(p))^{-1}\\
&=&\Phi(a\cdot_G\huaB_0(a))(p\cdot_H\huaB_1(p))\cdot_H\Phi(\huaB_0(a))\huaB_1(p)^{-1}\cdot_H\Phi(\huaB_0(a))\huaB_1(p)\\
&&\cdot_H\Phi(a\cdot_G\huaB_0(a))(q\cdot_H\huaB_1(q))\cdot_H\Phi(\huaB_0(a))\huaB_1(q)^{-1}\cdot_H\Phi(\huaB_0(a))\huaB_1(p)^{-1}\\
&=&\Phi(a\cdot_G\huaB_0(a))(p\cdot_H\huaB_1(p))\cdot_H\Phi(a\cdot_G\huaB_0(a))(q\cdot_H\huaB_1(q))\\
&&\cdot_H\Phi(\huaB_0(a))\huaB_1(q)^{-1}\cdot_H\Phi(\huaB_0(a))\huaB_1(p)^{-1},
\end{eqnarray*}
thus $\hat{\Phi}(a)(p\cdot_{\huaB_1}q)=\hat{\Phi}(a)(p)\cdot_{\huaB_1}\hat{\Phi}(a)(q)$.

On the other hand, since $(\huaB_1, \huaB_0)$ is a Rota-Baxter operator on $(H, G, t, \Phi)$, we have
\begin{eqnarray*}
&&\hat{\Phi}(a)\Big(\hat{\Phi}(b)p\Big)\\
&=&\Phi(a\cdot_G\huaB_0(a))\Big(\hat{\Phi}(b)(p)\cdot_H\Phi(\huaB_0(b))\huaB_1(p)\Big)\cdot_H\Phi(\huaB_0(a))\Big(\Phi(\huaB_0(b))\huaB_1(p)^{-1}\Big)\\
&=&\Phi(a\cdot_G\huaB_0(a))\Big(\Phi(b\cdot_H\huaB_0(b))(p\cdot_H\huaB_1(p))\Big)\cdot_H\Phi(\huaB_0(a)\cdot_G\huaB_0(b))\huaB_1(p)^{-1},
\end{eqnarray*}
and
\begin{eqnarray*}
&&\hat{\Phi}(a\cdot_{\huaB_0}b)p\\
&=&\Phi(a\cdot_G\huaB_0(a)\cdot_G b\cdot_G\huaB_0(b))(p\cdot_H\huaB_1(p))\cdot_H\Phi(\huaB_0(a)\cdot_G\huaB_0(b))\huaB_1(p)^{-1},
\end{eqnarray*}
which implies that $\hat{\Phi}(a)\Big(\hat{\Phi}(b)p\Big)=\hat{\Phi}(a\cdot_{\huaB_0}b)p$. Thus $\hat{\Phi}: G\lon\Aut(H)$ is a Lie group action of $(G, \cdot_{\huaB_0})$ on $(H, \cdot_{\huaB_1})$.

By \cite[Proposition 3.5]{JSZ}, we obtain that the inverse of $p$ in $(H, \cdot_{\huaB_1})$ is
$$
p^\dagger=\huaB_1(p)^{-1}\cdot_H p^{-1}\cdot_H\huaB_1(p).
$$
Since $t\circ\huaB_1=\huaB_0\circ t$, it follows that
\begin{eqnarray*}
t(\hat{\Phi}(a)p)&=&t\Big(\Phi(a\cdot_G\huaB_0(a))(p\cdot_H\huaB_1(p))\cdot_H\Phi(\huaB_0(a))\huaB_1(p)^{-1}\Big)\\
&=&a\cdot_G\huaB_0(a)\cdot_Gt(p\cdot_H\huaB_1(p))\cdot_G(a\cdot_G\huaB_0(a))^{-1}\cdot_G\huaB_0(a)\cdot_Gt(\huaB_1(p))^{-1}\cdot_G\huaB_0(a)^{-1}\\
&=&a\cdot_G\huaB_0(a)\cdot_Gt(p)\cdot_G\huaB_0(t(p))\cdot_G a^{-1}\cdot_G\huaB_0(t(p))^{-1}\cdot_G\huaB_0(a)^{-1}\\
&=&a\cdot_{\huaB_0} t(p)\cdot_{\huaB_0} a^{-1},
\end{eqnarray*}
and
\begin{eqnarray*}
&&\hat{\Phi}(t(p))q\\
&=&\Phi(t(p)\cdot_G\huaB_0(t(p)))(q\cdot_G\huaB_1(q))\cdot_H\Phi(\huaB_0(t(p)))\huaB_1(q)^{-1}\\
&=&(p\cdot_H\huaB_1(p))\cdot_H(q\cdot_G\huaB_1(q))\cdot_H (p\cdot_H\huaB_1(p))^{-1}\cdot_H\huaB_1(p)\cdot\huaB_1(q)^{-1}\cdot_G\huaB_1(p)^{-1} \\
&=&p\cdot_{\huaB_1} q\cdot_{\huaB_1} p^{\dagger}.
\end{eqnarray*}
Therefore, $\Big((H, \cdot_{\huaB_1}), (G, \cdot_{\huaB_0}), t, \hat{\Phi}\Big)$ is a crossed module of Lie groups.
\end{proof}

\begin{thm}\label{semirb}
Let $(\huaB_1, \huaB_0)$ be a Rota-Baxter operator on a crossed module of Lie groups $(H, G, t, \Phi)$. Then $\huaB=(\huaB_0, \huaD)$ is a Rota-Baxter operator on the semi-direct product Lie group $(G\ltimes_{\Phi} H, \cdot_{\Phi})$, where
$$
\huaD(a, p)=\Phi(\huaB_0(a))\huaB_1(\Phi(\huaB_0(a)^{-1}\cdot_G a^{-1})p), \quad \forall a\in G, p\in H.
$$
and
$$
(a, p)\cdot_\Phi(b, q)=(a\cdot_G b, p\cdot_H\Phi(a)q), \quad \forall a, b\in G, p, q\in H.
$$
\end{thm}
\begin{proof}
On the one hand, since $(\huaB_1, \huaB_0)$ is a Rota-Baxter operator on $(H, G, t, \Phi)$, we obtain
\begin{eqnarray*}
&&\huaB\Big((a, p)\cdot_{\Phi}\huaB(a, p)\cdot_{\Phi}(b, q)\cdot_{\Phi}\huaB(a, p)^{-1}\Big)\\
&=&\huaB\Big((a, p)\cdot_{\Phi}(\huaB_0(a), \huaD(a, p))\cdot_{\Phi}(b, q)\cdot_{\Phi}(\huaB_0(a), \huaD(a, p))^{-1}\Big)\\
&=&\huaB\Big((a, p)\cdot_{\Phi}(\huaB_0(a), \huaD(a, p))\cdot_{\Phi}(b, q)\cdot_{\Phi}(\huaB_0(a)^{-1},\Phi(\huaB_0(a)^{-1})\huaD(a, p)^{-1})\Big)\\
&=&\huaB\Big(a\cdot_G\huaB_0(a)\cdot_G b\cdot_G\huaB_0(a)^{-1},\\
&&p\cdot_H\Phi(a)\huaD(a, p)\cdot_H\Phi(a\cdot_G\huaB_0(a))q\cdot_H\Phi(a\cdot_G\huaB_0(a)\cdot_G b\cdot_G\huaB_0(a)^{-1})\huaD(a, p)^{-1}\Big)\\
&=&\huaB\Big(a\cdot_G\huaB_0(a)\cdot_G b\cdot_G\huaB_0(a)^{-1}, p\cdot_H\Phi(a\cdot_G\huaB_0(a))\huaB_1(\Phi(\huaB_0(a)^{-1}\cdot_G a^{-1})p)\cdot_H\\
&&\Phi(a\cdot_G\huaB_0(a))q\cdot_H\Phi(a\cdot_G\huaB_0(a)\cdot_G b)\huaB_1(\Phi(\huaB_0(a)^{-1}\cdot_G a^{-1})p)^{-1}\Big)\\
&=&\Big(\huaB_0(a)\cdot_G\huaB_0(b), \Phi(\huaB_0(a)\cdot_G\huaB_0(b))\\
&&\huaB_1\Big((\Phi(\huaB_0(b)^{-1}\cdot_G b^{-1}\cdot_G\huaB_0(a)^{-1}\cdot_G a^{-1})p\cdot_H\Phi(B_0(b)^{-1}\cdot_G b^{-1})\huaB_1(\Phi(\huaB_0(a)^{-1}\cdot_G a^{-1})p)\\
&&\cdot_H\Phi(\huaB_0(b)^{-1}\cdot_G b^{-1})q\cdot_H\Phi(\huaB_0(b)^{-1})\huaB_1(\Phi(\huaB_0(a)^{-1}\cdot_G a^{-1})p)^{-1})\Big)\\
&=&(\huaB_0(a)\cdot_G\huaB_0(b), e_H)\cdot_{\Phi}\\
&&\Big(e_G, \huaB_1\Big(\Phi(\huaB_0(b)^{-1}\cdot_G b^{-1}\cdot_G\huaB_0(a)^{-1}\cdot_G a^{-1})p\cdot_H\Phi(B_0(b)^{-1}\cdot_G b^{-1})\huaB_1(\Phi(\huaB_0(a)^{-1}\cdot_G a^{-1})p)\\
&&\cdot_H\Phi(\huaB_0(b)^{-1}\cdot_G b^{-1})q\cdot_H\Phi(\huaB_0(b)^{-1})\huaB_1(\Phi(\huaB_0(a)^{-1}\cdot_G a^{-1})p)^{-1}\Big)\Big).
\end{eqnarray*}
On the other hand, since $\huaB_0(b)^{-1}=\huaB_0(\huaB_0(b)^{-1}\cdot_G b^{-1}\cdot\huaB_0(b))$, we have
\begin{eqnarray*}
&&\huaB(a, p)\cdot_{\Phi}\huaB(b, q)\\
&=&\Big(\huaB_0(a)\cdot_G\huaB_0(b), \huaD(a, p)\cdot_H\Phi(\huaB_0(a))\huaD(b, q)\Big)\\
&=&\Big(\huaB_0(a)\cdot_G\huaB_0(b),\Phi(\huaB_0(a))\huaB_1(\Phi(\huaB_0(a)^{-1}\cdot_G a^{-1})p)\cdot_H\Phi(\huaB_0(a)\cdot_G\huaB_0(b))\huaB_1(\Phi(\huaB_0(b)^{-1}\cdot_G b^{-1})q)\Big)\\
&=&\Big(\huaB_0(a)\cdot_G\huaB_0(b), e_H\Big)\cdot_{\Phi}\Big(e_G, \Phi(\huaB_0(b)^{-1})\huaB_1(\Phi(\huaB_0(a)^{-1}\cdot_G a^{-1})p)\cdot_H\huaB_1(\Phi(\huaB_0(b)^{-1}\cdot_G b^{-1})q)\Big)\\
&=&\Big(\huaB_0(a)\cdot_G\huaB_0(b), e_H\Big)\cdot_{\Phi}\Big(e_G, \Phi(\huaB_0(\huaB_0(b)^{-1}\cdot_G b^{-1}\cdot_G\huaB_0(b)))\huaB_1(\Phi(\huaB_0(a)^{-1}\cdot_G a^{-1})p)\\
&&\cdot_H\huaB_1(\Phi(\huaB_0(b)^{-1}\cdot_G b^{-1})q)\Big)\\
&=&\Big(\huaB_0(a)\cdot_G\huaB_0(b), e_H\Big)\cdot_{\Phi}\Big(e_G, \huaB_1\Big(\Phi(\huaB_0(b)^{-1}\cdot_G b^{-1}\cdot_G\huaB_0(a)^{-1}\cdot_G a^{-1})p\\
&&\cdot_H\Phi(\huaB_0(b)^{-1}\cdot_G b^{-1})\huaB_1(\Phi(\huaB_0(a)^{-1}\cdot_G a^{-1})p)\cdot_H\Phi(\huaB_0(b)^{-1})\huaB_1(\Phi(\huaB_0(a)^{-1}\cdot_G a^{-1})p)^{-1}\Big)\\
&&\cdot_H\huaB_1(\Phi(\huaB_0(b)^{-1}\cdot_G b^{-1})q)\Big)\\
&=&\Big(\huaB_0(a)\cdot_G\huaB_0(b), e_H\Big)\cdot_{\Phi}\\
&&\Big(e_G, \huaB_1\Big(\Phi(\huaB_0(b)^{-1}\cdot_G b^{-1}\cdot_G\huaB_0(a)^{-1}\cdot_G a^{-1})p\cdot_H\Phi(B_0(b)^{-1}\cdot_G b^{-1})\huaB_1(\Phi(\huaB_0(a)^{-1}\cdot_G a^{-1})p)\\
&&\cdot_H\Phi(\huaB_0(b)^{-1}\cdot_G b^{-1})q\cdot_H\Phi(\huaB_0(b)^{-1})\huaB_1(\Phi(\huaB_0(a)^{-1}\cdot_G a^{-1})p)^{-1}\Big)\Big).
\end{eqnarray*}
Thus we have
$$
\huaB(a, p)\cdot_{\Phi}\huaB(b, q)=\huaB\Big((a, p)\cdot_{\Phi}\huaB(a, p)\cdot_{\Phi}(b, q)\cdot_{\Phi}\huaB(a, p)^{-1}\Big),
$$
which implies that $\huaB=(\huaB_0, \huaD)$ is a Rota-Baxter operator on $(G\ltimes_{\Phi} H, \cdot_{\Phi})$.
\end{proof}
As $\huaB=(\huaB_0, \huaD)$ is a Rota-Baxter operator on the semi-direct product Lie groups $(G\ltimes_{\Phi} H, \cdot_{\Phi})$, by Propoistion 3.5 in \cite{JSZ}, there is a Lie group $(G\times H, \cdot_{\huaB})$ induced by $(\huaB_0, \huaD)$, where
$$
(a, p)\cdot_{\huaB}(b, q)=(a, p)\cdot_{\Phi}\huaB(a, p)\cdot_{\Phi}(b, q)\cdot_{\Phi}\huaB(a, p)^{-1}, \quad \forall a, b\in G, p, q\in H.
$$
Moreover, it follows from Proposition \ref{dg} that $(H, \cdot_{\huaB_1})$ and $(G, \cdot_{\huaB_0})$ are Lie groups and $\hat{\Phi}: G\lon H$ is an action. Denote by $(G\ltimes_{\hat{\Phi}} H, \cdot_{\hat{\Phi}})$ the semi-direct product Lie group of $(G, \cdot_{\huaB_0})$ and $(H, \cdot_{\huaB_1})$, where
$$
(a, p)\cdot_{\hat{\Phi}}(b, q)=(a\cdot_{\huaB_0} b, p\cdot_{\huaB_1}\hat{\Phi}(a)q),\quad\forall a, b\in G, p, q\in H.
$$

Define a smooth map $\varphi: G\ltimes_{\hat{\Phi}} H\lon G\times H$ by
\begin{equation}\label{eqdefiD}
\varphi(a, p)=\Big(a, p\cdot_H\huaB_1(p)\cdot_H\Phi(a)\huaB_1(p)^{-1}\Big), \quad \forall a\in G, p\in H,
\end{equation}
then the Lie groups $(G\ltimes_{\hat{\Phi}} H, \cdot_{\hat{\Phi}})$ and $(G\times H, \cdot_\huaB)$ have the following relation.
\begin{pro}\label{isomdi}
With the above notations, the map $\varphi: G\ltimes_{\hat{\Phi}} H\lon G\times H$ is a Lie group isomorphism from $(G\ltimes_{\hat{\Phi}} H, \cdot_{\hat{\Phi}})$ to $(G\times H, \cdot_\huaB)$.
\end{pro}
\begin{proof}
For any $a\in G, p\in H$, by \eqref{eqdefiD}, we have
\begin{equation*}
\varphi(a, e_H)=(a, e_H), \quad \varphi(e_G, p)=(e_G, p).
\end{equation*}
By \eqref{deslie1}, \eqref{deslie2}, \eqref{deslie3} and \eqref{eqdefiD}, it follows that
\begin{eqnarray*}\label{eq0011}
&&\varphi((a, e_H)\cdot_{\hat{\Phi}}(b, q))\\
\nonumber&=&\varphi(a\cdot_{\huaB_{0}} b, \hat{\Phi}(a)q)\\
\nonumber&=&\Big(a\cdot_{\huaB_0} b, \hat{\Phi}(a)q\cdot_H\Phi(\huaB_0(a))\huaB_1(q)\cdot_H\Phi(a\cdot_G\huaB_0(a)\cdot_G b)\huaB_1(q)^{-1}\Big)\\
\nonumber&=&\Big(a\cdot_G\huaB_0(a)\cdot_G b\cdot_G\huaB_0(a)^{-1}, \Phi(a\cdot_G\huaB_0(a))(q\cdot_H\huaB_1(q))\cdot_H\Phi(a\cdot_G\huaB_0(a)\cdot_G b)\huaB_1(q)^{-1}\Big)\\
\nonumber&=&(a, e_H)\cdot_\huaB(b, q\cdot_H\huaB_1(q)\cdot_H\Phi(b)\huaB_1(q)^{-1})\\
\nonumber&=&\varphi(a, e_H)\cdot_{\huaB}\varphi(b, q),
\end{eqnarray*}
and
\begin{eqnarray*}\label{eq0022}
&&\varphi((e_G, p)\cdot_{\hat{\Phi}}(b, q))\\
\nonumber&=&\varphi(b, p\cdot_H\huaB_0(p)\cdot_H q\cdot_H\huaB_0(p)^{-1})\\
\nonumber&=&\Big(b, p\cdot_H\huaB_1(p)\cdot_Hq\cdot_H\huaB_1(q)\cdot_H\Phi(b)(\huaB_1(q)^{-1}\cdot_H\huaB_1(p)^{-1})\Big)\\
\nonumber&=&(e_G, p)\cdot_{\huaB}(b, q\cdot_H\huaB_1(q)\cdot_H\Phi(b)\huaB_1(q)^{-1})\\
\nonumber&=&\varphi(e_G, p)\cdot_{\huaB}\varphi(b, q).
\end{eqnarray*}
Denote by $a^\dag$ the inverse of $a$ in $(G, \cdot_{\huaB_0})$, we have
\begin{eqnarray*}
\varphi((a, p)\cdot_{\hat{\Phi}}(b, q))&=&\varphi\Big((a, e_H)\cdot_{\hat{\Phi}}(e_G, \hat{\Phi}(a^\dag)p)\cdot_{\hat{\Phi}}(b, e_H)\cdot_{\hat{\Phi}}(e_G, \hat{\Phi}(b^\dag)q)\Big)\\
&=&\varphi(a, e_H)\cdot_{\hat{\Phi}}\varphi\Big((e_G, \hat{\Phi}(a^\dag)p)\cdot_{\hat{\Phi}}(b, e_H)\cdot_{\hat{\Phi}}(e_G, \hat{\Phi}(b^\dag)q)\Big)\\
&=&\varphi(a, e_H)\cdot_{\hat{\Phi}}\varphi(e_G, \hat{\Phi}(a^\dag)p)\cdot_{\hat{\Phi}}\varphi\Big((b, e_H)\cdot_{\hat{\Phi}}(e_G, \hat{\Phi}(b^\dag)q)\Big)\\
&=&\varphi(a, e_H)\cdot_{\huaB}\varphi(e_G, \hat{\Phi}(a^\dag)p)\cdot_{\huaB}\varphi(b, e_H)\cdot_{\huaB}\varphi(e_G, \hat{\Phi}(b^\dag)q)\\
&=&\varphi((a, e_H)\cdot_{\hat{\Phi}}(e_G, \hat{\Phi}(a^\dag)p))\cdot_{\huaB}\varphi((b, e_H)\cdot_{\hat{\Phi}}(e_G, \hat{\Phi}(b^\dag)q))\\
&=&\varphi(a, p)\cdot_{\huaB}\varphi(b, q),
\end{eqnarray*}
which implies that $\varphi:G\ltimes_{\hat{\Phi}} H\lon G\times H$ is a Lie group homomorphism.

Moreover, it is obvious that $\ker\varphi=(e_G, e_H)$. For any $a\in G, p\in H$, we have $$\varphi(a, \hat{\Phi}(a)p)=\Big(a, \Phi(a\cdot_G\huaB_0(a))p\Big),$$ which implies that for any $(b, q)\in G\times H$, there is $$\varphi(b, \hat{\Phi}(b)\Phi(\huaB_0(b)^{-1}\cdot_G b^{-1})q)=(b, q).$$ Thus $\varphi: G\ltimes_{\hat{\Phi}} H\lon G\times H$ is a Lie group isomorphism from $(G\ltimes H, \cdot_{\hat{\Phi}})$ to $(G\times_{\hat{\Phi}} H, \cdot_\huaB)$.
\end{proof}

\section{Categorical solutions of the Yang-Baxter equation}

In this section, we introduce the notion of categorical solutions of the Yang-Baxter equation and this kind of solutions are constructed by Rota-Baxter operators on crossed modules of Lie groups.

Recall that a set-theoretical solution of the Yang-Baxter equation is a pair $(X, r)$, where X is a set and $r:X\times X\lon X\times X$ is a bijective map satisfying
\begin{equation}\label{eq-ybe}
(r\times\Id)(\Id\times r)(r\times\Id)=(\Id\times r)(r\times\Id)(\Id\times r).
\end{equation}
In generally, it's difficult to give a solution of the equation \eqref{eq-ybe}. There are some solutions of the equation \eqref{eq-ybe} constructed by bijective $1$-cocycles, matched pair of groups, skew braces and Rota-Baxter operators on Lie groups\cite{BG2, ESS, Ga, GM, GV, LYZ, Ru, Ta}.

Given Rota-Baxter operators on a Lie group, Bardakov and Gubarev constructed set-theoretical solutions of the Yang-Baxter equation as following.
\begin{thm}\label{yangbaxtersolufromrb}\cite{BG2}
Let $\huaB$ be a Rota-Baxter operator on a Lie group $G$. Then $$
R: G\times G\lon G\times G, \quad R(a, b)=(\Ad_{\huaB(a)}(b), \Ad_{(\huaB(\Ad_{\huaB(a)}(b))^{-1})(\Ad_{\huaB(a)}(b))^{-1}}a)
$$
is a solution of the Yang-Baxter equation on the set $G$.
\end{thm}

\begin{defi}
A categorical solution of the Yang-Baxter equation is a pair $(\huaC, R)$, where $\huaC$ is a small category and $R:\huaC\times\huaC\lon\huaC\times\huaC$ is \footnote{Here, $\huaC\times\huaC$ is the product category.}an invertible functor satisfying
$$
(R\times\Id)(\Id\times R)(R\times\Id)=(\Id\times R)(R\times\Id)(\Id\times R).
$$
\end{defi}
In the following, we will construct categorical solutions of the Yang-Baxter equation by Rota-Baxter operators on crossed modules of Lie groups.

Let $(H, G, t, \Phi)$ be a crossed module of Lie groups. As proven in \cite{BaL}, $(H, G, t, \Phi)$ gives rise to a small category $(\huaC, \frks, \frkt)$. More precisely,
\begin{equation}\label{deficate1}
\huaC_1=G\times H, \quad \huaC_0=G, \quad \frks(a, p)=a, \quad \frkt(a, p)=t(p)\cdot_G a, \quad \forall a\in G, p\in H,
\end{equation}
and the composition of morphisms is given by
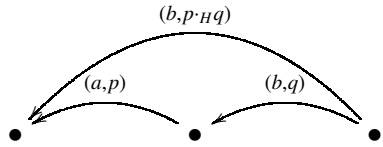
\begin{equation}\label{deficateg2}
(a, p)\circ(b, q)=(b, p\cdot_H q), \quad \text{where} \quad t(q)\cdot_G b=a.
\end{equation}

\[
\xymatrix{
  \bullet  && \bullet  \ar@/_1pc/[ll]_{(a,p)} && \bullet  \ar@/_1pc/[ll]_{(b,q)} \ar@/_3.2pc/[llll]_{(b,p\cdot_H q)}
}
\]
For any $a\in \huaC_0$, the identify morphism $1_a=(a, e_H)$.

Let $(\huaB_1, \huaB_0)$ be a Rota-Baxer operator on $(H, G, t, \Phi)$, we have the following result.
\begin{thm}\label{categorical}
Define two maps $R_0: \huaC_0\times\huaC_0\lon \huaC_0\times \huaC_0$ and $R_1:\huaC_1\times\huaC_1\lon\huaC_1\times\huaC_1$ by
\begin{eqnarray*}
R_0(c, d)&=&(\Ad_{\huaB_0(c)}d,  \Ad_{(\huaB_0(\Ad_{\huaB_0(c)}d)^{-1})(\Ad_{\huaB_0(c)}d)^{-1}}c), \\
R_1\Big((a, p), (b, q)\Big)&=&\Big(\Ad_{\huaB(a, p)}(b, q), \Ad_{(\huaB(\Ad_{\huaB(a, p)}(b, q))^{-1})(\Ad_{\huaB(a, p)}(b, q))^{-1}}(a, p)\Big),
\end{eqnarray*}
where
$$
\huaB(a, p)=\Big(\huaB_0(a), \Phi(\huaB_0(a))\huaB_1(\Phi(\huaB_0(a)^{-1}\cdot_G a^{-1})p)\Big), \quad \forall a\in G, p\in H.
$$
Then $R=(R_1, R_0):\huaC\times\huaC\lon\huaC\times\huaC$ is a categorical solution of the Yang-Baxter equation.
\end{thm}
\begin{proof}
We only need to prove that $R: \huaC\times\huaC\lon\huaC\times\huaC$ is a functor and $(\huaC_0, R_0), (\huaC_1, R_1)$ are both set-theoretical solutions of the Yang-Baxter equation.

For any $(e, x)\in G\times H$ and $(f, y)\in G\times H$, we can chose $a, b\in G$ and $p, q\in H$ such that
\begin{eqnarray}\label{001}
\left\{\begin{array}{rcl}
~~e&=&a, \\
~~x&=&\Phi(a\cdot_{G}\huaB_0(a))p, \\
~~f&=&\Ad_{\huaB_0(a)^{-1}}b, \\
~~y&=&\huaB_1(p)^{-1}\cdot_G\Phi(\huaB_0(a)^{-1}\cdot_G b\cdot_G\huaB_0(b))q\cdot_G\Phi(\Ad_{\huaB_0(a)^{-1}}b)\huaB_1(p)).
\end{array}\right.
\end{eqnarray}
Moreover, by the definition of $R_1$, it follows that
\begin{eqnarray*}
&&R_1\Big((a, \Phi(a\cdot_G\huaB_0(a))p), (\Ad_{\huaB_0(a)^{-1}}b, \huaB_1(p)^{-1}\cdot_G\Phi(\huaB_0(a)^{-1}\cdot_G b\cdot_G\huaB_0(b))q\cdot_G\\
&&\quad\quad\quad\quad\Phi(\Ad_{\huaB_0(a)^{-1}}b)\huaB_1(p))\Big)\\
&=&\Big((b, \Phi(b\cdot_G\huaB_0(b))q), \\
&&\quad(\Ad_{\huaB_0(b)^{-1}\cdot_G b^{-1}}a, \huaB_1(q)^{-1}\cdot_H q^{-1}\cdot_H\Phi(\huaB_0(b)^{-1}\cdot_H b^{-1}\cdot_G a\cdot_G\huaB_0(a))p\cdot_G\\
&&\quad\quad\quad\quad\quad\quad\Phi(\Ad_{\huaB_0(b)^{-1}\cdot_G b^{-1}}a)(q\cdot_H\huaB_1(q)))\Big).
\end{eqnarray*}

By \eqref{deficate1}, we have
\begin{eqnarray*}
&&(\frks\times\frks)R_1\Big((e, x), (f, y)\Big)\\
&=&(\frks\times\frks)R_1\Big((a, \Phi(a\cdot_G\huaB_0(a))p), \\
&&\quad\quad\quad\quad\quad(\Ad_{\huaB_0(a)^{-1}}b, \huaB_1(p)^{-1}\cdot_G\Phi(\huaB_0(a)^{-1}\cdot_G b\cdot_G\huaB_0(b))q\cdot_G\Phi(\Ad_{\huaB_0(a)^{-1}}b)\huaB_1(p))\Big)\\
&=&(b, \Ad_{\huaB_0(b)^{-1}\cdot_G b^{-1}}a)
\end{eqnarray*}
and
\begin{eqnarray*}
&&R_0((\frks\times\frks)\Big((e, x), (f, y)\Big))\\
&=&R_0((\frks\times\frks)\Big((a, \Phi(a\cdot_G\huaB_0(a))p), \\
&&\quad\quad\quad\quad\quad(\Ad_{\huaB_0(a)^{-1}}b, \huaB_1(p)^{-1}\cdot_G\Phi(\huaB_0(a)^{-1}\cdot_G b\cdot_G\huaB_0(b))q\cdot_G\Phi(\Ad_{\huaB_0(a)^{-1}}b)\huaB_1(p))\Big))\\
&=&(b, \Ad_{\huaB_0(b)^{-1}\cdot_G b^{-1}}a).
\end{eqnarray*}
Thus $(\frks\times\frks)R_1=R_0(\frks\times\frks)$.

By \eqref{crmo2}, we have
\begin{eqnarray*}
&&R_0\Big((\frkt\times\frkt)((e, x), (f, y))\Big)\\
&=& R_0\Big((\frkt\times\frkt)((a, \Phi(a\cdot_G\huaB_0(a))p)\\
 &&\quad\quad\quad (\Ad_{\huaB_0(a)^{-1}}b, \huaB_1(p)^{-1}\cdot_G\Phi(\huaB_0(a)^{-1}\cdot_G b\cdot_G\huaB_0(b))q\cdot_G\Phi(\Ad_{\huaB_0(a)^{-1}}b)\huaB_1(p)))\Big)\\
&=&R_0\Big(a\cdot_G\huaB_0(a)\cdot_G t(p)\cdot_G\huaB_0(a)^{-1}, \\
&&\quad\quad\quad t(\huaB_1(p))^{-1}\cdot_G\huaB_0(a)^{-1}\cdot_G b\cdot_G\huaB_0(b)\cdot_G t(q)\cdot_G\huaB_0(b)^{-1}\cdot_G\huaB_0(a)\cdot_G t(\huaB_1(p))\Big)\\
&=&\Big(b\cdot_G\huaB_0(b)\cdot_G t(q)\cdot_G\huaB_0(b)^{-1}, \Ad_{\huaB_0(t(q))^{-1}\cdot_G t(q)^{-1}\cdot_G\huaB_0(b)^{-1}\cdot_G b^{-1}}(a\cdot_G\huaB_0(a)\cdot_G t(p)\cdot_G\huaB_0(a)^{-1})\Big),
\end{eqnarray*}
and
\begin{eqnarray*}
&&(\frkt\times\frkt)R_1\Big((e, x), (f, y)\Big)\\
&=&(\frkt\times\frkt)R_1\Big((a, \Phi(a\cdot_G\huaB_0(a))p),\\
 &&\quad\quad\quad \quad (\Ad_{\huaB_0(a)^{-1}}b, \huaB_1(p)^{-1}\cdot_G\Phi(\huaB_0(a)^{-1}\cdot_G b\cdot_G\huaB_0(b))q\cdot_G\Phi(\Ad_{\huaB_0(a)^{-1}}b)\huaB_1(p))\Big)\\
&=&(\frkt\times\frkt)\Big((b, \Phi(b\cdot_G\huaB_0(b))q), \\
&&\quad\quad(\Ad_{\huaB_0(b)^{-1}\cdot_G b^{-1}}a, \huaB_1(q)^{-1}\cdot_H q^{-1}\cdot_H\Phi(\huaB_0(b)^{-1}\cdot_H b^{-1}\cdot_G a\cdot_G\huaB_0(a))p\cdot_G\\
&&\quad\quad\quad\quad\quad\Phi(\Ad_{\huaB_0(b)^{-1}\cdot_G b^{-1}}a)(q\cdot_H\huaB_1(q)))\Big)\\
&=&\Big(b\cdot_G\huaB_0(b)\cdot_G t(q)\cdot_G\huaB_0(b)^{-1},\Ad_{\huaB_0(t(q))^{-1}\cdot_G t(q)^{-1}\cdot_G\huaB_0(b)^{-1}\cdot_G b^{-1}}(a\cdot_G\huaB_0(a)\cdot_G t(p)\cdot_G\huaB_0(a)^{-1})\Big).
\end{eqnarray*}
Thus $(\frkt\times\frkt)R_1\Big((a, p), (b, q)\Big)=R_0\Big((\frkt\times\frkt)\Big((a, p), (b, q)\Big)\Big)$.

Moreover, for any $(a, b)\in \huaC_0\times\huaC_0$,
\begin{eqnarray*}
R_1(1_{(a, b)})&=&R_1(1_a, 1_b)=R_1((a, e_H), (b, e_H))\\
&=&\Big((\huaB_0(a)\cdot_G b\cdot_G\huaB_0(a)^{-1}, e_H), (\Ad_{(\huaB_0(\Ad_{\huaB_0(a)}b)^{-1})(\Ad_{\huaB_0(a)}b)^{-1}}a, e_H)\Big)\\
&=&1_{R_0(a, b)}.
\end{eqnarray*}

Let $(e', x')\in G\times H$ and $(f', y')\in G\times H$. There exists $(c, h)\in G\times H$ and $(d, k)\in G\times H$ such that
\begin{eqnarray}\label{002}
\left\{\begin{array}{rcl}
~~e'&=&c, \\
~~x'&=&\Phi(c\cdot_{G}\huaB_0(c))h, \\
~~f'&=&\Ad_{\huaB_0(c)^{-1}}d, \\
~~y'&=&\huaB_1(h)^{-1}\cdot_G\Phi(\huaB_0(c)^{-1}\cdot_G d\cdot_G\huaB_0(d))k\cdot_G\Phi(\Ad_{\huaB_0(c)^{-1}}d)\huaB_1(h)).
\end{array}\right.
\end{eqnarray}
Assume that $(\frkt\times\frkt)\Big((e', x'), (f', y')\Big)=(\frks\times\frks)\Big((e, x), (f, y)\Big)$, by \eqref{001} and \eqref{002}, we have
$$
c\cdot_G\huaB_0(c)\cdot_G t(h)\cdot_G\huaB_0(c)^{-1}=a, \quad d\cdot_G\huaB_0(d)\cdot_G t(k)\cdot_G\huaB_0(d)^{-1}=b.
$$
Denote by
\begin{equation*}
P=h\cdot_H\huaB_1(h)\cdot_H p\cdot_H\huaB_1(h)^{-1}, \quad Q=k\cdot_H\huaB_1(k)\cdot_H q\cdot_H\huaB_1(k)^{-1}.
\end{equation*}
By \eqref{deficateg2}, \eqref{001}, \eqref{002} and \eqref{crmo1}, it follows that
\begin{eqnarray*}
~(e, x)\circ (e', x')&=&(c, \Phi(c\cdot_G\huaB_0(c))P),\\
~(f, y)\circ (f', y')&=&\Big(\Ad_{\huaB_0(c)^{-1}}d, \huaB_1(p)^{-1}\cdot_H\Phi(\huaB_0(t(h))^{-1}\cdot_G\huaB_0(c)^{-1}\cdot_G d\cdot_G\huaB_0(d)\cdot_G t(k)\cdot_G\\
&&\huaB_0(t(k))q\cdot_H\Phi\Big(\Ad_{(\huaB_0(c)\cdot_G\huaB_0(t(h)))^{-1}}(d\cdot_G\huaB_0(d)\cdot_G t(k)\cdot_G\huaB_0(d)^{-1})\Big)\huaB_1(p)\\
&&\cdot_H\huaB_1(h)^{-1}\cdot_H\Phi(\huaB_0(c)^{-1}\cdot_G d\cdot_G\huaB_0(d))k\cdot_H\Phi(\Ad_{\huaB_0(c)^{-1}}d)B_1(h)\Big)\\
&=&\Big(\Ad_{\huaB_0(c)^{-1}}d, \huaB_1(P)^{-1}\cdot_G\Phi(\huaB_0(c)^{-1}\cdot_G d\cdot_G\huaB_0(d))Q\cdot_G\Phi(\Ad_{\huaB_0(c)^{-1}}d)\huaB_1(P))\Big).
\end{eqnarray*}
On the one hand, by \eqref{001} and \eqref{002},
\begin{eqnarray*}
&&R_1((e, x)\circ (e', x'), (f, y)\circ (f', y'))\\
&=&\Big((d, \Phi(d\cdot_G\huaB_0(d))Q), \Big(\Ad_{\huaB_0(d)^{-1}\cdot_G d^{-1}}c, \\
&&\huaB_1(Q)^{-1}\cdot_H Q^{-1}\cdot_H\Phi(\huaB_0(d)^{-1}\cdot_G d^{-1}\cdot_G c\cdot_G\huaB_0(c))P\cdot_H\Phi(\Ad_{\huaB_0(d)^{-1}\cdot_G d^{-1}}c)(Q\cdot_H\huaB_1(Q))\Big)\Big).
\end{eqnarray*}
On the other hand, by \eqref{001}, \eqref{002} and \eqref{crmo1},
\begin{eqnarray*}
&&R_1((e, x), (f, y))\circ R_1((e', x'), (f', y'))\\
&=&\Big((b, \Phi(b\cdot_G\huaB_0(b))q), \Big(\Ad_{\huaB_0(b)^{-1}\cdot_G b^{-1}}a, \\
&&\huaB_1(q)^{-1}\cdot_G q^{-1}\cdot_G\Phi(\huaB_0(b)^{-1}\cdot_G b^{-1}\cdot_G a\cdot_G\huaB_0(a))p\cdot_G\Phi(\Ad_{\huaB_0(b)^{-1}\cdot_G b^{-1}}a)(q\cdot_H\huaB_1(q))\Big)\Big)\\
&&\circ \Big((d, \Phi(d\cdot_G\huaB_0(d))k), \Big(\Ad_{\huaB_0(d)^{-1}\cdot_G d^{-1}}c, \\
&&\huaB_1(k)^{-1}\cdot_G k^{-1}\cdot_G\Phi(\huaB_0(d)^{-1}\cdot_G d^{-1}\cdot_G c\cdot_G\huaB_0(c))h\cdot_G\Phi(\Ad_{\huaB_0(d)^{-1}\cdot_G d^{-1}}c)(k\cdot_H\huaB_1(k))\Big)\Big)\\
&=&\Big((d, \Phi(d\cdot_G\huaB_0(d))Q),  \Big(\Ad_{\huaB_0(d)^{-1}\cdot_G d^{-1}}c, \\
&&\huaB_1(Q)^{-1}\cdot_H Q^{-1}\cdot_H\Phi(\huaB_0(d)^{-1}\cdot_G d^{-1}\cdot_G c\cdot_G\huaB_0(c))P\cdot_H\Phi(\Ad_{\huaB_0(d)^{-1}\cdot_G d^{-1}}c)(Q\cdot_H\huaB_1(Q))\Big)\Big).
\end{eqnarray*}
Thus $$R_1((e, x), (f, y))\circ R_1((e', x'), (f', y'))=R_1((e, x)\circ (e', x'), (f, y)\circ (f', y')), $$
when $$(\frkt\times\frkt)\Big((e', x'), (f', y')\Big)=(\frks\times\frks)\Big((e, x), (f, y)\Big).$$
We prove that $R=(R_1, R_0):\huaC\times\huaC\lon\huaC\times\huaC$ is a functor.

Since $(\huaC_1, \cdot_{\Phi})$ is a semi-direct product Lie group of $(G, \cdot_G)$ and $(H, \cdot_H)$, where
 $$
 (a, p)\cdot_{\Phi}(b, q)=(a\cdot_G b, p\cdot_G\Phi(a)q), \quad \forall (a, p), (b, q)\in\huaC_1,
 $$
by Theorem \ref{semirb}, $\huaB$ is a Rota-Baxter operator on $(\huaC_1, \cdot_\Phi)$. Due to $\huaB$ and $\huaB_0$ are Rota-Baxter operators on $\huaC_1$ and $\huaC_0$, by Theorem \ref{yangbaxtersolufromrb}, we have that $(\huaC_1, R_1)$ and $(\huaC_0, R_0)$ are both set-theoretical solution of the Yang-Baxter equation. Thus $R=(R_1, R_0):\huaC\times\huaC\lon\huaC\times\huaC$ is a categorical solution of the the Yang-Baxter equation.
\end{proof}

Moreover, we obtain another functor $F:\huaC\lon\huaC$ as following.
\begin{pro}
 With the above notations, $F=(\huaB, \huaB_0)$ is a functor from $\huaC$ to $\huaC$, where
$$
\huaB(a, p)=\Big(\huaB_0(a), \Phi(\huaB_0(a))\huaB_1(\Phi(\huaB_0(a)^{-1}\cdot_G a^{-1})p)\Big), \quad \forall a\in G, p\in H,
$$
\end{pro}
\begin{proof}
For any $a\in G$ and $p\in H$, since $(\huaB_1, \huaB_0)$ is a Rota-Baxter operator on $(H, G, t, \Phi)$, then
\begin{eqnarray*}
\frkt\Big(\huaB(a, p)\Big)&=&t\Big(\Phi(\huaB_0(a))\huaB_1(\Phi(\huaB_0(a)^{-1}\cdot_G a^{-1})p)\Big)\cdot_G\huaB_0(a)\\
&=&\huaB_0(a)\cdot_G t(\huaB_1(\Phi(\huaB_0(a)^{-1}\cdot_G a^{-1})p))\\
&=&\huaB_0(a)\cdot_G \huaB_0(t(\Phi(\huaB_0(a)^{-1}\cdot_G a^{-1})p))\\
&=&\huaB_0(t(p)\cdot_G a)\\
&=&\huaB_0(\frkt(a, p)),
\end{eqnarray*}
and
\begin{equation*}
\frks\Big(\huaB(a, p)\Big)=\huaB_0(a)=\huaB_0(\frks (a, p)).
\end{equation*}
Moreover, for any $(a, h)\in G\times H$ and $(b, k)\in G\times H$, there exists $p, q\in H$ such that
\begin{equation}\label{1234}
h=\Phi(a\cdot_{G}\huaB_0(a))p, \quad k=\Phi(b\cdot_{G}\huaB_0(b))q.
\end{equation}
Assume that $\frkt(b, k)=\frks(a, h)$, by \eqref{1234} and \eqref{crmo2},  it follows that $$t(q)=\huaB_0(b)^{-1}\cdot_G b^{-1}\cdot_G a\cdot_G\huaB_0(b).$$ On the one hand,
\begin{eqnarray*}
&&\huaB(a, h)\circ\huaB(b, k)\\
&=&\huaB(a, \Phi(a\cdot_G\huaB_0(a))p)\circ\huaB(b, \Phi(b\cdot_G\huaB_0(b))q)\\
&=&\Big(\huaB_0(a), \Phi(\huaB_0(a))\huaB_1(p)\Big)\circ\Big(\huaB_0(b), \Phi(\huaB_0(b))\huaB_1(q)\Big)\\
&=&\Big(\huaB_0(b), \Phi(\huaB_0(a))\huaB_1(p)\cdot_H\Phi(\huaB_0(b))\huaB_1(q)\Big)\\
&=&\Big(\huaB_0(b), \Phi(\huaB_0(b)\cdot_G\huaB_0(t(q)))\huaB_1(p)\cdot_H\Phi(\huaB_0(b))\huaB_1(q)\Big)\\
&=&\Big(\huaB_0(b), \Phi(\huaB_0(b))(\huaB_1(q)\cdot_H\huaB_1(p))\Big).
\end{eqnarray*}
On the other hand,
\begin{eqnarray*}
&&\huaB((a, h)\circ(b, k))\\
&=&\huaB((a, \Phi(a\cdot_G\huaB_0(a))p)\circ(b, \Phi(b\cdot_G\huaB_0(b))q))\\
&=&\huaB(b, \Phi(a\cdot_G\huaB_0(a))p\cdot_H\Phi(b\cdot_G\huaB_0(b))q)\\
&=&\Big(\huaB_0(b), \Phi(\huaB_0(b))\huaB_1(\Phi(t(q\cdot_H\huaB_1(q)))p\cdot_H q)\Big)\\
&=&\Big(\huaB_0(b), \Phi(\huaB_0(b))\huaB_1(q\cdot_H\huaB_1(q)\cdot_Hp\cdot_H\huaB_1(q)^{-1})\Big)\\
&=&\Big(\huaB_0(b), \Phi(\huaB_0(b))(\huaB_1(q)\cdot_H\huaB_1(p))\Big).
\end{eqnarray*}
Thus $\huaB(a, h)\circ\huaB(b, k)=\huaB((a, h)\circ(b, k)),$
  when $\frkt(b, k)=\frks(a, h)$.

Moreover, for any $a\in G$, one has $\huaB(1_a)=\huaB(a, e_H)=(\huaB_0(a), e_H)=1_{\huaB_0(a)}$. Thus $\huaB$ is a functor from $\huaC$ to $\huaC$.
\end{proof}

\section{Rota-Baxter operators on crossed modules of Lie algebras and its applications}

In this section, first, we recall Rota-Baxter operators on Lie algebras and crossed modules of Lie algebras. Then the notion of Rota-Baxter operators on crossed modules of Lie algebras is defined, and its properties are established.

\begin{defi}\rm(\cite{BGN, BD})
Let $(\g, [\cdot, \cdot]_\g)$ be a Lie algebra. A linear map $B:\g\lon\g$ is called a Rota-Baxter operator of weight $1$ if
\begin{equation*}
[B(x), B(y)]_\g=B([B(x), y]_\g+[x, B(y)]_\g+[x, y]_\g), \quad \forall x, y\in\g.
\end{equation*}
\end{defi}
\begin{rmk}
In the sequel, we call Rota-Baxter operators of weight $1$ simply Rota-Baxter operators without ambiguity.
\end{rmk}

\begin{ex}
Let $(\g, [\cdot, \cdot]_\g)$ be a Lie algebra. Then $B=-\Id$ is a Rota-Baxter operator on the Lie algebra $(\g, [\cdot, \cdot]_\g)$.
\end{ex}

\begin{defi}
Let $B$ and $B'$ be Rota-Baxters on Lie algebras $(\g, [\cdot, \cdot]_\g)$ and $(\h, [\cdot,\cdot]_\h)$ respectively. A homomorphism from $B'$ to $B$ is a Lie algebra homomorphism $\psi:\h\lon\g$ such that $\psi\circ B'=B\circ\psi$.
\end{defi}
\begin{defi}
A crossed module of Lie algebras is a quadruple $(\h, \g, \dM t, \phi)$, where $(\h, [\cdot, \cdot]_\h)$ and $(\g, [\cdot, \cdot]_\g)$ are Lie algebras, $\dM t:\h\lon\g$ is a Lie algebra homomorphism and $\phi:\g\lon\Der(\h)$ is a Lie algebra homomorphism, such that
\begin{equation*}
\dM t(\phi(x)u)=[x, \dM t(u)]_\g, \quad \phi(\dM t(u))v=[u, v]_\h, \quad \forall x\in\g, ~~u, v\in\h.
\end{equation*}
\end{defi}

Let $(\h, \g, \dM t, \phi)$ be a crossed module of Lie algebras. We consider semi-direct product Lie algebras $(\g\ltimes\g, [\cdot, \cdot]_{\g\ltimes\g})$ and $(\h\ltimes\h, [\cdot, \cdot]_{\h\ltimes\h})$, where
$$
\g\ltimes\g=\g\oplus\g, \quad \h\ltimes\h=\h\oplus\h,
$$
and
\begin{eqnarray}
\label{defi-g}[(x_1, y_1), (x_2, y_2)]_{\g\ltimes\g}&=&([x_1, x_2]_{\g}, [x_1, y_2]_{\g}+[y_1, x_2]_{\g}+[y_1, y_2]_\g), \quad \forall x_1, x_2, y_1, y_2\in\g,\\
\label{defi-h}~[(u_1, v_1), (u_2, v_2)]_{\h\ltimes\h}&=&([u_1, u_2]_{\h}, [u_1, v_2]_{\h}+[v_1, u_2]_{\h}+[v_1, v_2]_\h), \quad \forall u_1, u_2, v_1, v_2\in\h.
\end{eqnarray}
Define a linear map $\tilde{\phi}:\g\ltimes\g\lon\gl(\h\ltimes\h)$ by
\begin{equation}\label{eq-rep}
\tilde{\phi}(x, y)(u, v)=(\phi(x)u, \phi(y)u+\phi(y)v+\phi(x)v), \quad \forall x, y\in\g, ~~u, v\in\h.
\end{equation}
Then we have the following conclusion.

\begin{pro}\label{semicross}
With the above notations, $(\h\ltimes\h, \g\ltimes\g, (\dM t, \dM t), \tilde{\phi})$ is a crossed module of Lie algebras.
\end{pro}
\begin{proof}
We only proof that $\tilde{\phi}:\g\ltimes\g\lon\Der(\h\ltimes\h)$ is a Lie algebra homomorphism and leave other to interested readers. By \eqref{eq-rep}, for any $x, y\in\g, u_1, u_2, v_1, v_2\in\h$, we have
\begin{eqnarray*}
&&\tilde{\phi}(x, y)[(u_1, v_1), (u_2, v_2)]_{\h\ltimes\h}\\
&=&\tilde{\phi}(x, y)([u_1, u_2]_\h, [u_1, v_2]_\h+[v_1, u_2]_\h+[v_1, v_2]_\h)\\
&=&\Big(\phi(x)[u_1, u_2]_\h, \phi(y)[u_1, u_2]_\h+\phi(y)[u_1, v_2]_\h+\phi(y)[v_1, u_2]_\h+\phi(y)[v_1, v_2]_\h\\
&&+\phi(x)[u_1, v_2]_\h+\phi(x)[v_1, u_2]_\h+\phi(x)[v_1, v_2]_\h\Big)\\
&=&\Big([\phi(x)u_1, u_2]_\h, [\phi(y)u_1, u_2]_\h+[\phi(y)u_1, v_2]_\h+[\phi(y)v_1, u_2]_\h+[\phi(y)v_1, v_2]_\h\\
&&+[\phi(x)u_1, v_2]_\h+[\phi(x)v_1, u_2]_\h+[\phi(x)v_1, v_2]_\h\Big)\\
&&+\Big([u_1, \phi(x)u_2]_\h, [u_1, \phi(y)u_2]_\h+[u_1, \phi(y)v_2]_\h+[v_1, \phi(y)u_2]_\h+[v_1, \phi(y)v_2]_\h\\
&&+[u_1, \phi(x)v_2]_\h+[v_1, \phi(x)u_2]_\h+[v_1, \phi(x)v_2]_\h\Big)\\
&=&[(\phi(x)u_1, \phi(y)u_1+\phi(y)v_1+\phi(x)v_1), (u_2, v_2)]_{\h\ltimes\h}\\
&&+[(u_1, u_2), (\phi(x)u_2, \phi(y)u_2+\phi(y)v_2+\phi(x)v_2)]_{\h\ltimes\h}\\
&=&[\tilde{\phi}(x, y)(u_1, v_1), (u_2, v_2)]_{\h\ltimes\h}+[(u_1, v_1), \tilde{\phi}(x, y)(u_2, v_2)]_{\h\ltimes\h}.
\end{eqnarray*}
Moreover, by \eqref{defi-g}, for any $x_1, x_2, y_1, y_2\in\g, u, v\in\h$, we obtain
\begin{eqnarray*}
&&\tilde{\phi}([(x_1, y_1), (x_2, y_2)]_{\g\ltimes\g})(u, v)\\
&=&\tilde{\phi}([x_1, x_2]_\g, [x_1, y_2]_\g+[y_1, x_2]_\g+[y_1, y_2]_\g)(u, v)\\
&=&\Big(\phi([x_1, x_2]_\g)u, \phi([x_1, x_2]_\g)v+\phi([x_1, y_2]_\g)v+\phi([y_1, x_2]_\g)v+\phi([y_1, y_2]_\g)v\\
&&+\phi([x_1, y_2]_\g)u+\phi([y_1, x_2]_\g)u+\phi([y_1, y_2]_\g)u\Big)\\
&=&\Big(\phi(x_1)\phi(x_2)u, \phi(x_1)\phi(x_2)v+\phi(x_1)\phi(y_2)u+\phi(x_1)\phi(y_2)v\\
&&+\phi(y_1)\phi(x_2)u+\phi(y_1)\phi(x_2)v+\phi(y_1)\phi(y_2)u+\phi(y_1)\phi(y_2)v\Big)\\
&&-\Big(\phi(x_2)\phi(x_1)u, \phi(x_2)\phi(x_1)v+\phi(y_2)\phi(x_1)u+\phi(y_2)\phi(x_1)v\\
&&+\phi(x_2)\phi(y_1)u+\phi(x_2)\phi(y_1)v+\phi(y_2)\phi(y_1)u+\phi(y_2)\phi(y_1)v\Big)\\
&=&\tilde{\phi}(x_1, y_1)\tilde{\phi}(x_2, y_2)(u, v)-\tilde{\phi}(x_2, y_2)\tilde{\phi}(x_1, y_1)(u, v)\\
&=&[\tilde{\phi}(x_1, y_1), \tilde{\phi}(x_2, y_2)](u, v).
\end{eqnarray*}
Thus $\tilde{\phi}:\g\ltimes\g\lon\Der(\h\ltimes\h)$ is a Lie algebra homomorphism.
\emptycomment{
By \eqref{defi-h}, for any $u_1, u_2, v_1, v_2\in\h$,  we have
\begin{eqnarray*}
&&(\dM t, \dM t)([(u_1, v_1), (u_2, v_2)]_{\h\ltimes\h})\\
&=&(\dM t, \dM t)([u_1, u_2]_\h, [u_1, v_2]_\h+[v_1, u_2]_\h+[v_1, v_2]_\h)\\
&=&([\dM t(u_1), \dM t(u_2)]_\h, [\dM t(u_1), \dM t(v_2)]_\h+[\dM t(v_1), \dM t(u_2)]_\h+[\dM t(v_1), \dM t(v_2)]_\h)\\
&=&[(\dM t(u_1), \dM t(v_1)), (\dM t(u_2), \dM t(v_2))]_{\h\ltimes\h},
\end{eqnarray*}
which implies that $(\dM t, \dM t):\h\ltimes\h\lon\g\ltimes\g$ is a Lie algebra homomorphism.
Since $(\h, \g, \dM t, \phi)$ is a crossed module of Lie algebras, then we have
\begin{eqnarray*}
&&[(x, y), (\dM t(u), \dM t(v))]_{\g\ltimes\g}\\
&=&\Big([x, \dM t(u)]_\g, [y, \dM t(v)]_\g+[x, \dM t(v)]_\g+[y, \dM t(u)]_\g\Big)\\
&=&\Big(\dM t(\phi(x)u), \dM t(\phi(x)v)+\dM t(\phi(y)u)+\dM t(\phi(y)v)\Big)\\
&=&(\dM t, \dM t)(\phi(x)u, \phi(x)v+\phi(y)u+\phi(y)v)\\
&=&(\dM t, \dM t)\tilde{\Phi}(x, y)(u, v),
\end{eqnarray*}
and
\begin{eqnarray*}
&&\tilde{\phi}((\dM t, \dM t)(u_1, v_1))(u_2, v_2)\\
&=&\tilde{\phi}(\dM t(u_1), \dM t(v_1))(u_2, v_2)\\
&=&\Big(\phi(\dM t(u_1))u_2, \phi(\dM t(u_1))v_2+\phi(\dM t(v_1))u_2+\phi(\dM t(v_1))v_2\Big)\\
&=&\Big([u_1, u_2]_\h, [u_1, v_2]_\h+[v_1, u_2]_\h+[v_1, v_2]_\h\Big)\\
&=&[(u_1, u_2), (v_1, v_2)]_{\h\ltimes\h}.
\end{eqnarray*}
Therefore, the quadruple $(\h\ltimes\h, \g\ltimes\g, (\dM t, \dM t), \tilde{\phi})$ is a crossed module of Lie algebras.
}
\end{proof}

\begin{defi}
Let $(\h, \g, \dM t, \phi)$ be a crossed module of Lie algebras. A Rota-Baxter operator on $(\h, \g, \dM t, \phi)$ is a pair $(B_1, B_0)$, where linear maps $B_1: \h\lon\h$ and $B_0:\g\lon\g$ satisfy
\begin{itemize}
\item [\rm(i)] $B_1$ and $B_0$ are Rota-Baxter operators on $\h$ and $\g$ respectively,
\item [\rm(ii)]the linear map $\dM t$ is a homomorphism from $B_1$ to $B_0$, i.e. $\dM t\circ B_1=B_0\circ\dM t$,
\item [\rm(iii)] $\phi(B_0(x))B_1(u)=B_1\Big(\phi(x)B_1(u)+\phi(B_0(x))u+\phi(x)u\Big), \quad \forall x\in\g, u\in\h.$
\end{itemize}
\end{defi}

\begin{ex}
Let $(\h, \g, \dM t, \phi)$ be a crossed module of Lie algebras. Then $(-\Id, -\Id)$ is a Rota-Baxter operator on $(\h, \g, \dM t, \phi)$.
\end{ex}

\begin{ex}
Let $\g$ be a Lie algebra and $B$ be a Rota-Baxter operator on $\g$. Then $(B, B)$ is a Rota-Baxter operator on the crossed modules of Lie algebras $(\g, \g, \dM t, \ad)$, where $\dM t(x)=x$.
\end{ex}

In the following, we will character Rota-Baxter operators on crossed modules of Lie algebras as crossed module of Lie algebras.
\begin{thm}
Let $(\h, \g, \dM t, \phi)$ be a crossed module of Lie algebras and $B_1:\h\lon\h, B_0:\g\lon\g$ be linear maps. Then the pair $(B_1, B_0)$ is a Rota-Baxtr operator on $(\h, \g, \dM t, \phi)$ if and only if $(\mathrm{Gr}(B_1), \mathrm{Gr}(B_0), (\dM t, \dM t), \tilde{\phi})$ is a crossed module of Lie algebras, where
$$
\mathrm{Gr}(B_1)=\{(B_1(u), u)|\forall u\in\h\}\subset\h\ltimes\h, \quad \mathrm{Gr}(B_0)=\{(B_0(x), x)| \forall x\in\g\}\subset\g\ltimes\g.
$$
$$
(\dM t, \dM t): \mathrm{Gr}(B_1)\lon \mathrm{Gr}(B_0), \quad (\dM t, \dM t)(\huaB_1(u), u)=(\dM t(\huaB_1(u)), \dM t(u)),
$$
and $\tilde{\phi}:\mathrm{Gr}(B_0)\lon\Der(\mathrm{Gr}(B_1))$ defined by
$$
\tilde{\phi}(x, y)(u, v)=(\phi(x)u, \phi(y)u+\phi(y)v+\phi(x)v), \quad \forall x, y\in\g, ~~u, v\in\h.
$$
\end{thm}
\begin{proof}
On the one hand, suppose that $(B_1, B_0)$ is a Rota-Baxter operator on $(\h, \g, \dM t, \phi)$, by \cite[Proposition 2.5]{JSZ}, it follows that $\mathrm{Gr}(B_1), \mathrm{Gr}(B_0)$ are Lie subalgebras of $(\h\ltimes\h, [\cdot, \cdot]_{\h\ltimes\h})$ and $(\g\ltimes\g, [\cdot, \cdot]_{\g\ltimes\g})$ respectively.

By Proposition \ref{semicross}, $(\dM t, \dM t)$ is a Lie algebra homomorphism from $(\h\ltimes\h, [\cdot, \cdot]_{\h\ltimes\h})$ to $(\g\ltimes\g, [\cdot, \cdot]_{\g\ltimes\g})$, we obtain
$$\forall u\in\h, \quad (\dM t, \dM t)(B_1(u), u)=(\dM_t(B_1(u)), \dM t(u))=(B_0(\dM t(u)), \dM t(u))\in \mathrm{Gr}(B_0),$$
which implies that $(\dM t, \dM t)$ is a Lie algebra homomorphism from $\mathrm{Gr}(B_1)$ to $\mathrm{Gr}(B_0)$.

Since $(B_1, B_0)$ is a Rota-Baxter operator on $(\h, \g, \dM t, \phi)$, then
\begin{eqnarray*}
&&\tilde{\phi}(B_0(x), x)(B_1(u), u)\\
&=&(\phi(B_0(x))B_1(u), \phi(B_0(x))u+\phi(x)B_1(u)+\phi(x)u)\\
&=&\Big(B_1(\phi(B_0(x))u+\phi(x)B_1(u)+\phi(x)u), \phi(B_0(x))u+\phi(x)B_1(u)+\phi(x)u\Big)\in\mathrm{Gr}(B_1),
\end{eqnarray*}
thus $\tilde{\phi}$ is an action of $\mathrm{Gr}(B_0)$ on $\mathrm{Gr}(B_1)$ by Proposition \ref{semicross}.

As $(\h, \g, \dM t, \phi)$ is a crossed module of Lie algebras, we have
\begin{eqnarray*}
\tilde{\phi}\Big((\dM t, \dM t)(B_1(u), u)\Big)(B_1(v), v)&=&(\phi\Big(\dM t(B_1(u))\Big)B_1(v), \phi(\dM t(u))B_1(v)+\phi(\dM t(u))v+\phi\Big(\dM t(B_1(u))\Big)v)\\
&=&([B_1(u), B_1(v)]_\h, [u, B_1(v)]_\h+[u, v]_\h+[B_1(u), v]_\h)\\
&=&[(B_1(u), u), (B_1(v), v)]_{\h\ltimes\h},
\end{eqnarray*}
and
\begin{eqnarray*}
(\dM t, \dM t)\Big(\tilde{\phi}(B_0(x), x)(B_1(u), u)\Big)&=&(\dM t, \dM t)\Big(\phi(B_0(x))B_1(u), \phi(B_0(x))u+\phi(x)B_1(u)+\phi(x)u\Big)\\
&=&([B_0(x), B_0(\dM t(u))]_\g, [B_0(x), \dM t(u)]_\g+[x, \dM t(B_1(u))]_\g+[x, \dM t(u)]_\g)\\
&=&[(B_0(x), x), (B_0(\dM t(u)), \dM t(u))]_{\h\ltimes\h}\\
&=&[(B_0(x), x), (\dM t(B_1(u)), \dM t(u))]_{\h\ltimes\h}.
\end{eqnarray*}
Thus $(\mathrm{Gr}(B_1), \mathrm{Gr}(B_0), (\dM t, \dM t), \tilde{\phi})$ is a crossed module of Lie algebras.

On the other hand, suppose that $(\mathrm{Gr}(B_1), \mathrm{Gr}(B_0), (\dM t, \dM t), \tilde{\phi})$ is a crossed module of Lie algebras, by \cite[Proposition 2.5]{JSZ}, we obtain that $B_1:\h\lon\h, B_0:\g\lon\g$ are Rota-Baxter operators.

Moreover, by the fact that $(\dM t, \dM t): \mathrm{Gr}(B_1)\lon\mathrm{Gr}(B_0)$ is a Lie algebra homomorphism, it follows $\dM t\circ B_1=B_0\circ\dM t$.

Since $\tilde{\phi}(\mathrm{Gr}(B_0))\mathrm{Gr}(B_1)\subseteq\mathrm{Gr}(B_1)$, then
\begin{eqnarray*}
\tilde{\phi}(B_0(x), x)(B_1(u), u)
&=&\Big(\phi(B_0(x))B_1(u), \phi(B_0(x))u+\phi(x)B_1(u)+\phi(x)u\Big)\in\mathrm{Gr}(B_1).
\end{eqnarray*}
which implies $\phi(B_0(x))B_1(u)=B_1\Big(\phi(B_0(x))u+\phi(x)B_1(u)+\phi(x)u\Big)$. Thus $(B_1, B_0)$ is a Rota-Baxter operator on $(\h, \g, \dM t, \phi)$.
\end{proof}

Parallel to Proposition \ref{dg}, Rota-Baxter operators on crossed modules of Lie algebras descend crossed modules of Lie algebras.
\begin{pro}\label{promath}
Let $(B_1, B_0)$ be a Rota-Baxter operator on a crossed module of Lie algebras $(\h, \g, \dM t, \phi)$. Then $\Big((\h, [\cdot, \cdot]_{B_1}), (\g, [\cdot, \cdot]_{B_0}), \dM t, \hat{\phi}\Big)$ is a crossed module of Lie algebras, where
\begin{eqnarray*}
~~[u, v]_{B_1}&=&[B_1(u), v]_\h+[u, B_1(v)]_\h+[u, v]_\h, \quad \forall u, v\in\h,\\
~~[x, y]_{B_0}&=&[B_0(x), y]_\g+[x, B_0(y)]_\g+[x, y]_\g, \quad \forall x, y\in\g,\\
~~\hat{\phi}(x)u&=&\phi(x)B_1(u)+\phi(B_0(x))u+\phi(x)u, \quad \forall x\in\g, ~~u\in\h.
\end{eqnarray*}
\end{pro}
\begin{proof}
By \cite[Corollary 2.5]{JSZ}, $(\g, [\cdot, \cdot]_{B_0})$ and $(\h, [\cdot, \cdot]_{B_1})$ are Lie algebras. By the fact that $\dM t\circ B_1=B_0\circ\dM t$ and $\dM t:(\h, [\cdot, \cdot]_\h)\lon(\g, [\cdot, \cdot]_\g)$ is a Lie algebra homomorphism, it follows that $\dM t: (\h, [\cdot, \cdot]_{B_1})\lon(\g, [\cdot, \cdot]_{B_0})$ is a Lie algebra homomorphism.

For $x, y\in\g, u\in\h$, since $\phi(B_0(x))B_1(u)=B_1\Big(\phi(x)B_1(u)+\phi(B_0(x))u+\phi(x)u\Big)$, then
\begin{eqnarray*}
\hat{\phi}([x, y]_{B_0})u&=&\phi([B_0(x), y]_\g+[x, B_0(y)]_\g+[x, y]_\g)B_1(u)+\phi([B_0(x), B_0(y)]_\g)u\\
&&+\phi([B_0(x), y]_\g+[x, B_0(y)]_\g+[x, y]_\g)u\\
&=&\hat{\phi}(x)\hat{\phi}(y)u-\hat{\phi}(y)\hat{\phi}(x)u\\
&=&[\hat{\phi}(x), \hat{\phi}(y)]u,
\end{eqnarray*}
which implies that $\hat{\phi}:\g\lon\gl(\h)$ is a Lie algebra homomorphism.

Since $\phi(B_0(x))B_1(u)=B_1\Big(\phi(x)B_1(u)+\phi(B_0(x))u+\phi(x)u\Big)$, we obtain
\begin{eqnarray*}
&&[\hat{\phi}(x)u, v]_{B_1}+[u, \hat{\phi}(x)v]_{B_1}\\
&=&[B_1(\hat{\phi}(x)u), v]_\h+[\hat{\phi}(x)u, B_1(v)]_{\h}+[B_1(u), \hat{\phi}(x)v]_{\h}+[u, B_1(\hat{\phi}(x)v)]_\h+[\hat{\phi}(x)u, v]_{\h}+[u, \hat{\phi}(x)v]_{\h}\\
&=&[\phi(B_0(x))B_1(u), v]_\h+[u, \phi(B_0(x))B_1(v)]_\h+[\phi(x)B_1(u)+\phi(B_0(x))u+\phi(x)u, B_1(v)]_\h\\
&&+[B_1(u), \phi(x)B_1(v)+\phi(B_0(x))v+\phi(x)v]_\h+[\phi(x)B_1(u)+\phi(B_0(x))u+\phi(x)u, v]_\h\\
&&+[u, \phi(x)B_1(v)+\phi(B_0(x))v+\phi(x)v]_\h\\
&=&\phi(x)[B(u), B(v)]_\h+\phi(B_0(x))[u, v]_{B_1}+\phi(x)[u, v]_{B_1}\\
&=&\hat{\phi}(x)([u, v]_{B_1}),
\end{eqnarray*}
thus $\hat{\phi}$ is a Lie algebra action of $(\g, [\cdot, \cdot]_{B_0})$ on $(\h, [\cdot, \cdot]_{B_1})$.

Moreover, since $\dM t\circ B_1=B_0\circ\dM t$, it follows that
$$
\hat{\phi}(\dM t(u))v=\phi(\dM t(u))B_1(v)+\phi(B_0(\dM t(u)))v+\phi(\dM t(u))v=[u, v]_{B_1},
$$
and
\begin{eqnarray*}
\dM t(\hat{\phi}(x)u)&=&\dM t(\phi(x)B_1(u)+\phi(B_0(x))u+\phi(x)u)\\
&=&[x, B_0(\dM t(u))]_\g+[B_0(x), \dM t(u)]_\g+[x, \dM t(u)]_\g=[x, \dM t(u)]_{B_0}.
\end{eqnarray*}
Therefore, $\Big((\h, [\cdot, \cdot]_{B_1}), (\g, [\cdot, \cdot]_{B_0}), \dM t, \hat{\phi}\Big)$ is a crossed module of Lie algebras.
\end{proof}

\begin{pro}
Let $(\h, \g, \dM t, \phi)$ be a crossed module of Lie algebras and $(B_1, B_0)$ be a Rota-Baxter operator on $(\h, \g, \dM t, \phi)$. Then $B=(B_0, B_1)$ is a Rota-Baxter operator on the semi-direct product Lie algebra $(\g\ltimes\h, [\cdot, \cdot]_{\phi})$ of $(\g, [\cdot, \cdot]_\g)$ and $(\h, [\cdot, \cdot]_\h)$, where
$$
B(x, u)=(B_0(x), B_1(u)), \quad [(x, u), (y, v)]_{\phi}=([x, y]_\g, \phi(x)v-\phi(y)u+[u, v]_\h), \quad \forall x, y\in\g, u, v\in\h.
$$
\end{pro}
\begin{proof}
By the fact that $(B_1, B_0)$ is a Rota-Baxter operator on $(\h, \g, \dM t, \phi)$, we have
\begin{eqnarray*}
&&[B(x, u), B(y, v)]_{\phi}\\
&=&([B_0(x), B_0(y)]_\g, \phi(B_0(x))B_1(v)-\phi(B_0(y))B_1(u)+[B_1(u), B_1(v)]_\h)\\
&=&(B_0([B_0(x), y]_\g), B_1(\phi(B_0(x))v-\phi(y)B_1(u)+[u, v]_\h))\\
&&+(B_0([x, B_0(y)]_\g), B_1(\phi(x)B_1(v)-\phi(B_0(y))u+[u, B_1(v)]_\h))\\
&&+(B_0([x, y]_\g), B_1(\phi(x)v-\phi(y)u+[u, v]_\h))\\
&=&B\Big([B(x, u), (y, v)]_{\phi}+[(x, u), B(y, v)]_{\phi}+[(x, u), (y, v)]_{\phi}\Big),
\end{eqnarray*}
which implies that $B=(B_0, B_1)$ be a Rota-Baxter operator on $(\g\ltimes\h, [\cdot, \cdot]_{\phi})$.
\end{proof}

From Proposition \ref{promath}, we obtain that $(\g, [\cdot, \cdot]_{B_0}), (\h, [\cdot, \cdot]_{B_1})$ are Lie algebras and $\hat{\phi}:\g\lon\Der(\h)$ is a Lie algebra homomorphism. Denote by $(\g\ltimes\h, [\cdot, \cdot]_{\hat{\phi}})$ the semi-direct product Lie algebra of $(\g, [\cdot, \cdot]_{B_0})$ and $(\h, [\cdot, \cdot]_{B_1})$, where
$$
~[(x, u), (y, v)]_{\hat{\phi}}=([x, y]_{B_0}, \hat{\phi}(x)v-\hat{\phi}(y)u+[u, v]_{B_1}).
$$
Moreover, due to $(B_0, B_1)$ is a Rota-Baxter operator on $(\g\ltimes\h, [\cdot, \cdot]_{\phi})$, it follows that $(\g\oplus\h, [\cdot, \cdot]_{(B_0, B_1)})$ is a Lie algebra by Corollary 2.5 in \cite{JSZ}, where
\begin{eqnarray*}
~[(x, u), (y, v)]_{(B_0, B_1)}&=&[B(x, u), (y, v)]_{\ltimes_{\phi}}+[(x, u), B(y, v)]_{\ltimes_{\phi}}+[(x, u), (y, v)]_{\ltimes_{\phi}},
\end{eqnarray*}
Indeed, for any $(x, u), (y, v)\in\g\ltimes\h$, one has
\begin{eqnarray*}
&&[(x, u), (y, v)]_{(B_0, B_1)}\\
&=&([B_0(x), y]_\g, \phi(B_0(x))v-\phi(y)B_1(u)+[B_1(u), v]_{\h})\\
&&+([x, B_0(y)]_\g, \phi(x)B_1(v)-\phi(B_0(y))u+[u, B_1(v)]_{\h})+([x, y]_\g, \phi(x)v-\phi(y)u+[u, v]_\h)\\
&=&([x, y]_{B_0}, \hat{\phi}(x)v-\hat{\phi}(y)u+[u, v]_{B_1})=[(x, u), (y, v)]_{\hat{\phi}}.
\end{eqnarray*}
Then we have the following result.
\begin{pro}
The Lie algebra $(\g\oplus\h, [\cdot, \cdot]_{(B_0, B_1)})$ is same as the semi-direct product Lie algebra $(\g\ltimes\h, [\cdot, \cdot]_{\hat{\phi}})$.
\end{pro}

In \cite{BGN, Bor}, they provided a geometric explanation of the Rota-Baxter operator. More precisely, let $B$ be a Rota-Baxter operator on a Lie algebra $(\g, [\cdot,\cdot]_\g)$. Then $(\g, [\cdot, \cdot]_B)$ is a Lie algebra, where
$$
[x, y]_B=[B(x), y]_\g+[x, B(y)]_\g+[x, y]_\g, \quad \forall x, y\in\g.
$$
Let $(K, \cdot_K)$ be a connected and simply connected Lie group such that its Lie algebra be $(\g, [\cdot, \cdot]_B)$. According to \cite{Kob}, a left invariant connection $\nabla$ on $K$ is determined by $\nabla_XY(e_K)$ for all left invariant vector fields $X, Y$. There is a left invariant connection $\nabla^{B}$ on $K$ constructed as
\begin{equation}\label{con}
\nabla^{B}_XY(e_G)=[B(x), y]_\g+\frac{1}{2}[x, y]_\g,
\end{equation}
where $X, Y$ are left invariant vector fields and $X(e_G)=x, Y(e_G)=y$. Furthermore, $\nabla$ is a torsion free connection on $K$.

\begin{defi}\cite{PSV}
Let $\nabla^{M}$ and $\nabla^{N}$ be connections in $TM$ and $TN$ respectively, and let $F:M\lon N$ be a smooth map. $\nabla^{M}$ and $\nabla^{N}$ are called $F$-related if for any $X_1, X_2\in\frak X(M)$ and $Y_1, Y_2\in\frak X(N)$ such that $X_i$ and $Y_i$ are $F$-related, $i=1, 2$, then $\nabla^{M}_{X_1}X_2$ and $\nabla^{N}_{Y_1}Y_2$ are $F$-related.
\end{defi}

Let $(B_1, B_0)$ be a Rota-Baxter operator on a crossed module of Lie algebras $(\h, \g, \dM t, \phi)$. It follows that $(\g, [\cdot, \cdot]_{B_0}), (\h, [\cdot, \cdot]_{B_1})$ are Lie algebras and $\dM t: (\h, [\cdot, \cdot]_{B_1})\lon (\g, [\cdot, \cdot]_{B_0})$ is a Lie algebra homomorphism by Proposition \ref{promath}. Denote by $(G, \star)$ and $(H, \circ)$ the Lie groups corresponding with $(\g, [\cdot, \cdot]_{B_0})$ and $(\h, [\cdot, \cdot]_{B_1})$ respectively. Denote by $t:(H, \circ)\lon(G, \star)$ the Lie group homomorphism corresponding with $\dM t: (\h, [\cdot, \cdot]_{B_1})\lon (\g, [\cdot, \cdot]_{B_0})$. Define left invariant connections $\nabla^{B_1}$ and $\nabla^{B_0}$ on $H$ and $G$ by
$$
(\nabla^{B_1}_{X}Y)(e_H)=[B_1(u), v]_\h+\frac{1}{2}[u, v]_\h, \quad (\nabla^{B_0}_{\tilde{X}}\tilde{Y})(e_G)=[B_0(x), y]_\g+\frac{1}{2}[x, y]_\g,
$$
where $X, Y$ are left invariant vector fields on $(H, \circ)$ and $X(e_H)=u\in\h, Y(e_H)=v\in\h$, $\tilde{X}, \tilde{Y}$ are left invariant vector fields on $(G, \star)$ and $\tilde{X}(e_G)=x\in\g, \tilde{Y}(e_G)=y\in\g$.
\begin{pro}\label{2-connection}
The connections $\nabla^{B_1}$ and $\nabla^{B_0}$ are $t$-related.
\end{pro}
\begin{proof}
As the tangent bundles $TH$ and $TG$ are trivial bundle, in order to prove that $\nabla^{B_1}$ and $\nabla^{B_0}$ are $t$-related, we only need to prove that when left invariant vector fields $X, \tilde{X}$ are $t$-related and $Y, \tilde{Y}$ are $t$-related, then $\nabla^{B_1}_XY, \nabla^{B_0}_{\tilde{X}}\tilde{Y}$ are $t$-related.

Since $X(e_H)=x, Y(e_H)=y$ and $\tilde{X}(e_G)=\dM t(x), \tilde{Y}(e_G)=\dM t(y)$, it follows
\begin{eqnarray*}
t_{*a}(\nabla^{B_1}_XY(a))&=&t_{*a}(L_a)_{*e_H}\nabla^{B_1}_XY(e_H)=(tL_{a})_{*e_H}([B_1(x), y]_\h+\frac{1}{2}[x, y]_\h)\\
&=&(L_{t(a)})_{*e_G}t_{*e_H}([B_1(x), y]_\h+\frac{1}{2}[x, y]_\h)\\
&=&(L_{t(a)})_{*e_G}([B_0(\dM t(x)), \dM t(y)]_\g+\frac{1}{2}[\dM t(x), \dM t(y)]_\g)=(L_{t(a)})_{*e_G}(\nabla^{B_0}_{\tilde{X}}\tilde{Y}(e_G))\\
&=&\nabla^{B_0}_{\tilde{X}}\tilde{Y}(t(a)),
\end{eqnarray*}
thus $\nabla^{B_1}$ and $\nabla^{B_0}$ are $t$-related.
\end{proof}

\section{Lie theory for Rota-Baxter operators on crossed modules of Lie groups and Rota-Baxter operators on crossed modules of Lie algebras}
In this section, we study the differential of Rota-Baxter operators on crossed of Lie groups and the integration of Rota-Baxter operators on crossed modules of Lie algebras.

It is well known that there is one-to-one correspondence between finite dimensional Lie algebras and Lie groups. For crossed modules of Lie groups and crossed modules of Lie algebras, there is the same theorem as following.

\begin{thm}\rm(\cite{WF})
Let $\g$ and $\h$ be finite dimensional Lie algebras and integrate to connected, simply-connected Lie groups $G, H$. There is one-to-one correspondence between crossed modules of $\g$ and $\h$ and crossed modules of $G$ and $H$.
\end{thm}
Let $(H, G, t, \Phi)$ be a crossed module of Lie groups. By Proposition \ref{croslie}, we obtain that $(H\ltimes H, G\ltimes G, (t, t), \tilde{\Phi})$ is a crossed of Lie groups, where $H\ltimes H, G\ltimes G$ are semi-direct product Lie groups, and
\begin{equation*}
\tilde{\Phi}\Big((a, b)\Big)(p, q)=\Big(\Phi(a)p, \Phi(b\cdot_G a)(q\cdot_H p)\cdot_H\Phi(a)p^{-1}\Big), \quad \forall a, b\in G, p, q\in H.
\end{equation*}
Let $(\h, \g, \dM t, \phi)$ be a crossed module of Lie algebras of $(H, G, t, \Phi)$. By Proposition \ref{semicross}, $(\h\ltimes\h, \g\ltimes\g, (\dM t, \dM t), \tilde{\phi})$ is a crossed module of Lie algebras, where $\h\ltimes\h, \g\ltimes\g$ are semi-direct product Lie algebras and
\begin{equation*}
\tilde{\phi}(x, y)(u, v)=(\phi(x)u, \phi(y)u+\phi(y)v+\phi(x)v), \quad \forall x, y\in\g, ~~u, v\in\h.
\end{equation*}
\begin{pro}\label{thmcrlg}
With the above notations, $(\h\ltimes\h, \g\ltimes\g, (\dM t, \dM t), \tilde{\phi})$ is the crossed module of Lie algebras of $(H\ltimes H, G\ltimes G, (t, t), \tilde{\Phi})$.
\end{pro}
\begin{proof}
Since $\g$ and $\h$ are Lie algebras of $G$ and $H$, it follows that $\g\ltimes\g$ and $\h\ltimes\h$ are Lie algebras of $G\ltimes G$ and $H\ltimes H$. Moreover, for the semi-direct Lie group $H\ltimes H$, the exponential map $\exp_{\h\ltimes\h}:\h\ltimes\h\lon H\ltimes H$ is $$\exp_{\h\ltimes\h}(u, v)=\Big(\exp_\h(u), P_2(\exp_{\h\ltimes\h}(u, v))\Big),$$ where $P_2: H\ltimes H\lon H$ is defined by $P_2(p, q)=q$. By the fact that $t_*=\dM t$, we have
\begin{equation*}
(t, t)_*(u, v)=\frac{d}{ds}|_{s=0}\Big(t(\exp_\h(su)), t(P_2(\exp_{\h\ltimes\h}(su, sv)))\Big)=(\dM t(u), \dM t(v)),
\end{equation*}
thus $(t, t)_*=(\dM t, \dM t)$.

Since $\tilde{\Phi}:G\ltimes G\lon \Aut(H\ltimes H)$ is a Lie group homomorphism, then we obtain a Lie group homomorphism $\hat{\Phi}:G\ltimes G\lon \Aut(\h\ltimes\h)$, given by
$$
\hat{\Phi}((a, b))(u, v)=\frac{d}{ds}|_{s=0}\tilde{\Phi}((a, b))\exp_{\h\ltimes\h}(su, sv), \quad \forall a, b\in G, u, v\in\h.
$$
Taking the differential, there is a Lie algebra homomorphism $\tilde{\Phi}_*: \g\ltimes\g\lon\Der(\h\ltimes\h)$, given by
$$
\tilde{\Phi}_*((x, y))(u, v)=\frac{d}{ds}|_{s=0}\hat{\Phi}(\exp_{\g\ltimes\g}(sx, sy))(u, v),\quad \forall x, y\in\g, u, v\in\h.
$$

Similarly, define $\bar{\Phi}: G\lon\Aut(\h)$ and $\Phi_*: \g\lon\Der(\h)$ by
$$
\bar{\Phi}(a)u=\frac{d}{ds}|_{s=0}\Phi(a)\exp_\h(su), \quad \Phi_*(x)u=\frac{d}{ds}|_{s=0}\bar{\Phi}(\exp(sx))u.
$$
For $x, y\in\g, u, v\in\h$, by \eqref{eqdefrG} and $\Phi_*=\phi$, it follows
\begin{eqnarray*}
&&\tilde{\Phi}_*(x, y)(u, v)\\
&=&\frac{d}{ds_1}\frac{d}{ds_2}|_{s_1=s_2=0}\tilde{\Phi}(\exp_{\g\ltimes\g}(s_1x, s_1y))\exp_{\h\ltimes\h}(s_2u, s_2v)\\
&=&\frac{d}{ds_1}\frac{d}{ds_2}|_{s_1=s_2=0}\tilde{\Phi}\Big((\exp_\g(s_1x), P_2(\exp_{\g\ltimes\g}(s_1x, s_1y)))\Big)\Big(\exp_\h(s_2u), P_2(\exp_{\h\ltimes\h}(s_2u, s_2v))\Big)\\
&=&\frac{d}{ds_1}\frac{d}{ds_2}|_{s_1=s_2=0}\Big(\Phi(\exp_\g(s_1x))\exp(s_2u),\Phi\Big(P_2(\exp_{\g\ltimes\g}(s_1x, s_1y))\cdot_G\exp_\g(s_1x)\Big)\\
&&\Big(P_2(\exp_{\h\ltimes\h}(s_2u, s_2v))\cdot_H\exp_\h(s_2u)\Big)\cdot_H\Phi(\exp_\g(s_1x))\exp_\h(-s_2u)\Big)\\
&=&\frac{d}{ds_1}|_{s_1=0}\Big(\bar{\Phi}(\exp_\g(s_1x))u, \bar{\Phi}\Big(P_2(\exp_{\g\ltimes\g}(s_1x, s_1y))\cdot_G\exp_\g(s_1x)\Big)(v+u)-\bar{\Phi}(\exp_\g(s_1x))u\Big)\\
&=&\Big(\Phi_*(x)u, \Phi_*(x)v+\Phi_*(y)v+\Phi_*(y)u\Big)\\
&=&(\phi(x)u, \phi(x)v+\phi(y)v+\phi(y)u)\\
&=&\tilde{\phi}(x, y)(u, v).
\end{eqnarray*}
Therefore, we have that $(\h\ltimes\h, \g\ltimes\g, (\dM t, \dM t), \tilde{\phi})$ is the crossed module of Lie algebras of $(H\ltimes H, G\ltimes G, (t, t), \tilde{\Phi})$.
\end{proof}

In \cite{GLS}, they proved that tangent maps of Rota-Baxter operators on a Lie groups are Rota-Baxter operators on their Lie algebras. Conversely, Rota-Baxter operators on Lie algebras can be locally integrated into Rota-Baxter operators on corresponding Lie groups\cite{JSZ}, global integrability of Rota-Baxter operators on Lie algebras have obstructions\cite{JSZ1}. For Rota-Baxter operators on crossed modules of Lie algebras and Rota-Baxter operators on crossed modules of Lie groups, there are the following relations.

\begin{thm}\rm(Differentiation)\label{d1}
Let $(H, G, t, \Phi)$ be a crossed module of Lie groups. Denote by $(\h, \g, \dM t, \phi)$ the crossed module of Lie algebras of $(H, G, t, \Phi)$. If $(\huaB_1, \huaB_0)$ is a Rota-Baxter operator on $(H, G, t, \Phi)$, then $(B_1, B_0)$ is a Rota-Baxter operator on $(\h, \g, \dM t, \phi)$, where $B_1=(\huaB_1)_{*e_H}$ and $B_0=(\huaB_0)_{*e_G}$.
\end{thm}
\begin{proof}
Since $\huaB_1$ and $\huaB_0$ are Rota-Baxter operators on $H$ and $G$, by \cite[Theorem 2.9]{GLS}, we have that $(\huaB_1)_{*e_H}=B_1, (\huaB_0)_{*e_G}=B_0$ are Rota-Baxter operators on $\h, \g$. Moreover, since $t\circ\huaB_1=\huaB_0\circ t$, it follows $\dM t\circ B_1=B_0\circ\dM t$.

For any $x\in\g, u\in\h$, we have
\begin{eqnarray*}
&&\phi(B_0(x))B_1(u)\\
&=&\frac{d}{ds_1}\frac{d}{ds_2}|_{s_1=s_2=0}\Phi(\huaB_0(\exp_G(s_1x)))\huaB_1(\exp_H(s_2u))\\
&=&\frac{d}{ds_1}\frac{d}{ds_2}|_{s_1=s_2=0}\huaB_1\Big(\Phi(\exp_G(s_1x)\cdot_G\huaB_0(\exp_G(s_1x)))(\exp_H(s_2u)\cdot_H\huaB_1(\exp_H(s_2u)))\\
&&\cdot_H\Phi(\huaB_0(\exp_G(s_1x)))\huaB_1(\exp_H(s_2u))^{-1}\Big)\\
&=&B_1\Big(\frac{d}{ds_1}\frac{d}{ds_2}|_{s_1=s_2=0}\Phi(\exp_G(s_1x)\cdot_G\huaB_0(\exp_G(s_1x)))(\exp_H(s_2u)\cdot_H\huaB_1(\exp_H(s_2u)))\\
&&+\frac{d}{ds_1}\frac{d}{ds_2}|_{s_1=s_2=0}\Phi(\huaB_0(\exp_G(s_1x)))\huaB_1(\exp_H(s_2u))^{-1}\Big)\\
&=&B_1(\phi(x)B_1(u)+\phi(x)u+\phi(B_0(x))u).
\end{eqnarray*}
Thus $(B_1, B_0)$ is a Rota-Baxter operator on $(\h, \g, \dM t, \phi)$.
\end{proof}

Let $(\huaB_1, \huaB_0)$ be a Rota-Baxter operator on $(H, G, t, \Phi)$. By Proposition \ref{dg}, we obtain another crossed module of Lie groups $\Big((H, \cdot_{\huaB_1}), (G, \cdot_{\huaB_0}), t, \hat{\Phi}\Big)$, where
\begin{eqnarray*}
~p\cdot_{\huaB_1}q&=&p\cdot_H\huaB_1(p)\cdot_H q\cdot_H\huaB_1(p)^{-1}, \quad \forall p, q\in H,\\
~a\cdot_{\huaB_0}b&=&a\cdot_G\huaB_0(a)\cdot_G b\cdot_G\huaB_0(a)^{-1}, \quad \forall a, b\in G,\\
~\hat{\Phi}(a)p&=&\Phi(a\cdot_G\huaB_0(a))(p\cdot_H\huaB_1(p))\cdot_H\Phi(\huaB_0(a))\huaB_1(p)^{-1}, \quad \forall a\in G, p\in H.
\end{eqnarray*}
Denote by $(\h, \g, \dM t, \phi)$ the crossed module of Lie algebras of $(H, G, t, \Phi)$. Let $(\huaB_1)_{*e_H}=B_1$ and $(\huaB_0)_{*e_G}=B_0$. Then by Theorem \ref{d1} and Proposition \ref{promath}, $\Big((\h, [\cdot, \cdot]_{B_1}), (\g, [\cdot, \cdot]_{B_0}), \dM t, \hat{\phi}\Big)$ is a crossed module of Lie algebras, where
\begin{eqnarray*}
~~[u, v]_{B_1}&=&[B_1(u), v]_\h+[u, B_1(v)]_\h+[u, v]_\h, \quad \forall u, v\in\h,\\
~~[x, y]_{B_0}&=&[B_0(x), y]_\g+[x, B_0(y)]_\g+[x, y]_\g, \quad \forall x, y\in\g,\\
~~\hat{\phi}(x)u&=&\phi(x)B_1(u)+\phi(B_0(x))u+\phi(x)u, \quad \forall x\in\g, ~~u\in\h.
\end{eqnarray*}

\begin{pro}\label{d2}
With the above notations, $\Big((\h, [\cdot, \cdot]_{B_1}), (\g, [\cdot, \cdot]_{B_0}), \dM t, \hat{\phi}\Big)$ is the crossed module of Lie algebras of $\Big((H, \cdot_{\huaB_1}), (G, \cdot_{\huaB_0}), t, \hat{\Phi}\Big)$.
\end{pro}
\begin{proof}
By \cite[Proposition 4.5]{JSZ}, it follows that $(\h, [\cdot, \cdot]_{B_1})$ and $(\g, [\cdot, \cdot]_{B_0})$ are Lie algebras of $(H, \cdot_{\huaB_1}), (G, \cdot_{\huaB_0})$ respectively. Since
\begin{eqnarray*}
&&\frac{d}{dt}\frac{d}{ds}|_{t=0, s=0}\hat{\Phi}(\exp(tx))\exp(su)\\
&=&\frac{d}{dt}\frac{d}{ds}|_{t=0, s=0}\Phi(\exp(tx)\cdot_{\huaB_0}\huaB_0(\exp(tx)))(\exp(su)\cdot_{\huaB_1}\huaB_1(\exp(su)))\cdot_{\huaB_1}\Phi(\huaB_0(\exp(tx)))\huaB_1(\exp(su))^{-1}\\
&=&\phi(x)u+\phi(B_0(x))u+\phi(x)B_1(u)\\
&=&\hat{\phi}(x)u,
\end{eqnarray*}
thus $\Big((\h, [\cdot, \cdot]_{B_1}), (\g, [\cdot, \cdot]_{B_0}), \dM t, \hat{\phi}\Big)$ is a crossed module of Lie algebras corresponding with a crossed module of Lie groups $\Big((H, \cdot_{\huaB_1}), (G, \cdot_{\huaB_0}), t, \hat{\Phi}\Big)$.
\end{proof}

Finally, we study the integration of Rota-Baxter operators on crossed modules of Lie algebras.
\begin{thm}\rm(Integration)\label{i1}
Let $(\h, \g, \dM t, \phi)$ be a crossed module of Lie algebras. Denote by $(H, G, t, \Phi)$ the crossed module of Lie groups corresponding with $(\h, \g, \dM t, \phi)$, where $G$ and $H$ are connected, simply connected Lie groups. If $(B_1, B_0)$ is a Rota-Baxter operator on $(\h, \g, \dM t, \phi)$, and $\huaB_1, \huaB_0$ are Rota-Baxter operators on $H, G$ such that $(\huaB_1)_{*e_H}=B_1, (\huaB_0)_{*e_G}=B_0$, then $(\huaB_1, \huaB_0)$ is a Rota-Baxter operator on $(H, G, t, \Phi)$.
\end{thm}
\begin{proof}
Since $(\h, \g, \dM t, \phi)$ is a crossed module of Lie algebras and corresponding crossed module of Lie groups is$(H, G, t, \Phi)$, by Proposition \ref{semicross}, Proposition \ref{croslie} and Proposition \ref{thmcrlg}, we have that $(H\ltimes H, G\ltimes G, (t, t), \tilde{\Phi})$ is a crossed module of Lie groups and its corresponding crossed module of Lie algebras is $(\h\ltimes\h, \g\ltimes\g, (\dM t, \dM t), \tilde{\phi})$. Since $(\huaB_1)_{*e_H}=B_1, (\huaB_0)_{*e_G}=B_0$, by \cite[Proposition 4.5]{JSZ}, we have that $\mathrm{Gr}(B_1)$ and $\mathrm{Gr}(B_0)$ are Lie algebras of $\mathrm{Gr}(\huaB_1)$ and $\mathrm{Gr}(\huaB_0)$.

Moreover, the map $\tilde{\phi}:\g\ltimes\g\lon\Der(\h\ltimes\h)$ is a Lie algebra homomorphism and satisfies $$\tilde{\phi}(\mathrm{Gr}(B_0))\subseteq \Der_{\mathrm{Gr}(B_1)}(\h\ltimes\h),$$
where the set $\Der_{\mathrm{Gr}(B_1)}(\h\ltimes\h)=\{\delta\in\Der(\h\ltimes\h)|\delta(\mathrm{Gr}(B_1))\subseteq\mathrm{Gr}(B_1)\}$.
Suppose that a map $\bar{\phi}:G\ltimes G\lon\Aut(\h\ltimes\h)$ is an integration of $\tilde{\phi}$, we have
\begin{equation}\label{eq1}
\bar{\phi}(\exp_{G\ltimes G}(x, y))=e^{\tilde{\phi}(x, y)}, \quad \forall (x, y)\in\mathrm{Gr}(B_0).
\end{equation}
Since $\mathrm{Gr}(B_0)$ is a Lie algebra of $\mathrm{Gr}(\huaB_0)$ and $\mathrm{Gr}(\huaB_0)$ is connected, by \eqref{eq1}, it follows that $$\bar{\phi}(\mathrm{Gr}(\huaB_0))\subseteq \Aut_{\mathrm{Gr}(B_1)}(\h\ltimes\h),$$
where the set $\Aut_{\mathrm{Gr}(B_1)}(\h\ltimes\h)=\{\Delta\in\Aut(\h\ltimes\h)|\Delta(\mathrm{Gr}(B_1))\subseteq\mathrm{Gr}(B_1)\}$.

Assume that $\iota:\Aut(\h\ltimes \h)\lon\Aut(H\ltimes H)$ is a Lie group isomorphism, then $\tilde{\Phi}=\iota\circ\bar{\phi}$. For $f\in\Aut_{\mathrm{Gr}(B_1)}(\h\ltimes\h)$, since $\mathrm{Gr}(\huaB_1)$ is connected and
$$
\iota(f)(\exp_{\h\ltimes\h}(u, v))=\exp_{\h\ltimes\h}(f(u, v))\in\mathrm{Gr}(\huaB_1),\quad \forall (u, v)\in\mathrm{Gr}(B_1),
$$
we have $\iota(f)\in\Aut_{\mathrm{Gr}(\huaB_1)}(H\ltimes H)$, where $$\Aut_{\mathrm{Gr}(\huaB_1)}(H\ltimes H)=\{\chi\in\Aut(H\ltimes H)|\chi(\mathrm{Gr}(\huaB_1))\subseteq\mathrm{Gr}(\huaB_1)\}.$$
Then it follows that $\tilde{\Phi}(\mathrm{Gr}(\huaB_0))=\iota(\hat{\phi}(\mathrm{Gr}(\huaB_0)))\subset\Aut_{\mathrm{Gr}(\huaB_1)}(H\ltimes H).$ That is, for any $a\in G, p\in H$,
$$
\tilde{\Phi}(\huaB_0(a), a)(\huaB_1(p), p)=\Big(\Phi(\huaB_0(a))\huaB_1(p), \Phi(a\cdot_G \huaB_0(a))(p\cdot_H \huaB_1(p))\cdot_H\Phi(\huaB_0(a))\huaB_1(p)^{-1}\Big)\in\mathrm{Gr}(\huaB_1).
$$
Thus we have
$$
\Phi(\huaB_0(a))\huaB_1(p)=\huaB_1\Big(\Phi(a\cdot_G\huaB_0(a))(p\cdot_H\huaB_1(p))\cdot_H\Phi(\huaB_0(a))\huaB_1(p)^{-1}\Big), \quad a\in G, p\in H,
$$
which implies that $(\huaB_1, \huaB_0)$ is a Rota-Baxter operator on $(H, G, t, \Phi)$.
\end{proof}
\begin{rmk}
Theorem \ref{i1} means that the integrability of Rota-Baxter operators $(B_1, B_0)$ on $(\h, \g, \dM t, \phi)$ are same as the integrability of Rota-Baxter operators $B_1$ and $B_0$ on $\h$ and $\g$.
\end{rmk}

\vspace{2mm}
\noindent
{\bf Acknowledgements}  This research is supported by NSFC(12401076), China postdoctoral Science Foundation (2023M741349).


\begin{thebibliography}{a}
\bibitem{Baez}
J.C. Baez and A.S. Crans, Higher-dimensional algebra. VI. Lie 2-algebras. \emph{Theory Appl. Categ.} {\bf 12} (2004), 492-538.

\bibitem{BaL}
J. C. Baez and A. D. Lauda, Higher-dimensional algebra. V. 2-groups. \emph{Theory Appl. Categ.} {\bf 12} (2004), 423-491.

\bibitem{BaezSch}
J. C. Baez and U. Schrediber, Higher Gauge Theory, in Categories in Algebra, Geometry and Mathematical Physics,  A. Davydov et al eds., \emph{Contemp. Math.} {\bf 431}, AMS, Providence, (2007), pp. 7-30.


\bibitem{BGN} C. Bai, L. Guo and X. Ni, Nonabelian generalized Lax pairs, the classical Yang-Baxter equation and PostLie algebras. {\em Comm. Math. Phys.} {\bf 297} (2010), 553-596.


\bibitem{BG2}
V. G. Bardakov and V. Gubarev, Rota-Baxter groups, skew left braces, and the Yang-Baxter equation. \emph{J. Algebra} {\bf 596} (2022), 328-351.

\bibitem{BNMY}
V. G. Bardakov, M. V. Neshchadim and M. K. Yadav, Symmetric skew braces and brace systems. \emph{Forum Math.} {\bf 35} (2023), no. 3, 713-738.

\bibitem{Bax}
G. Baxter, An analytic problem whose solution follows from a simple algebraic identity.
\emph{Pacific J. Math.} {\bf 10} (1960), 731-742.

\bibitem{BAX}
R. J. Baxter, Partition function of the eight-vertex lattice model. \emph{Ann. Phys.} {\bf70} (1972), 193-228.

\bibitem{BD}
A.A. Belavin, V.G. Drinfeld, Solutions of the classical Yang-Baxter equation for simple Lie algebras. \emph{Funct. Anal. Appl.} {\bf 16} (1982) 159-180.

\bibitem{Bor}
M. Bordemann, Generalized Lax Pairs, the modified classical Yang-Baxter equation, and affine geometry of Lie groups. \emph{Comm. Math. Phys.} {\bf 135} (1990), 201-216.


\bibitem{CS}
A. Caranti and L. Stefanello, Skew braces from Rota-Baxter operators: a cohomological characterisation and some examples. \emph{Ann. Mat. Pura Appl.(4)} {\bf 202} (2023), no. 1, 1-13.

\bibitem{Car}
P. Cartier, On the structure of free Baxter algebras. \emph{Adv. Math.}, {\bf 9} (1972), 253-265.



\bibitem{CK}
A. Connes and D. Kreimer, Renormalization in quantum field theory and the Riemann-Hilbert problem. I. The Hopf algebra structure of graphs and the main theorem. {\em  Comm. Math. Phys.} {\bf 210} (2000), 249-273.


\bibitem{Dr}
V. G. Drinfel¡¯d, On some unsolved problems in quantum group theory. \emph{Lecture Notes in Math}. {\bf 1510} (1992),
1-8.

\bibitem{ESS}
P. Etingof, T. Schedler and A. Soloviev, Set-theoretical solutions to the quantum Yang-Baxter equation. \emph{Duke Math. J.} {\bf 100} (1999), no.2, 169-209.

\bibitem{Gas}
A. Gastel, Canonical gauges in higher gauge theory. \emph{Comm. Math. Phys.} {\bf 376} (2020), no.2, 1053-1071.

\bibitem{Ga}
T. Gateva-Ivanova, Set-theoretic solutions of the Yang-Baxter equation, braces and symmetric groups. \emph{Adv. Math.} {\bf 338} (2018), 649-701.

\bibitem{GaC}
T. Gateva-Ivanova and P. Cameron, Multipermutation solutions of the
Yang-Baxter equation. \emph{Comm. Math. Phys.} {\bf 309} (2012), 583-621.

\bibitem{GM}
T. Gateva-Ivanova, S. Majid, Matched pairs approach to set theoretic solutions of the Yang-Baxter equation. \emph{J. Algebra} {\bf 319} (2008), no. 4, 1462-1529.


\bibitem{GV}
L. Guarnieri, L. Vendramin, Skew braces and the Yang-Baxter equation. \emph{Math. Comput.} {\bf 86} (2017), 2519-2534.

\bibitem{GLS}
L. Guo, H. Lang and Y. Sheng, Integration and geometrization of Rota-Baxter Lie algebras.  {\em Adv. Math.}  {\bf387} (2021), Paper No. 107834, 34 pp.



\bibitem{JSZ}
J. Jiang, Y. Sheng and C. Zhu, Lie theory and cohomology of relative Rota-Baxter operators. \emph{J. Lond. Math. Soc. (2)} {\bf 109} (2024), no. 2, Paper No. e12863, 34pp.

\bibitem{JSZ1}
J. Jiang, Y. Sheng and C. Zhu, On the integration of relative Rota-Baxter Lie algebras. arXiv:2410.23547.

\bibitem{Kob}
S. Kobayashi and K. Nomizu, Foundations of Differential Geometry. \emph{New York: Wiley,} {\bf Vol. I}, 1963 and
{\bf Vol. II}, 1969.

\bibitem{Kos}
Y. Kosmann-Schwarzbach and F. Magri, Poisson-Lie groups and complete integrability, I: drinfeld bialgebras, dual extensions and their canonical representations. \emph{Ann. Inst. H. Poincar\'e Phys. Theo.} {\bf 49} (1988), 433-460.

\bibitem{Ku}
B. A. Kupershmidt, What a classical $r$-matrix really is. \emph{J. Nonlinear Math. Phys.} {\bf 6} (1999), 448-488.

\bibitem{LS}
H. Lang and Y. Sheng, Factorizable Lie bialgebras, quadratic Rota-Baxter Lie algebras and Rota-Baxter Lie bialgebras. \emph{Comm. Math. Phys.} {\bf 397} (2023), no.2, 763-791.


\bibitem{LYZ}
J. Lu, M. Yan and Y. Zhu, On the set-theoretical Yang-Baxter equation. \emph{Duke. Math. J.} {\bf 104} (2000), no.1, 1-18.

\bibitem{MM}
J. F. Martins and A. Mikovic, Lie crossed modules and gauge-invariant actions for $2$-BF theories. \emph{Adv. Theor. Math. Phys.} {\bf 15} (2011), 1059-1084.

\bibitem{NM}
R. Nishant and S. Mahender, Relative Rota-Baxter groups and skew left braces. To appear in \emph{Forum Math.}, https://doi.org/10.1515/forum-2024-0020.

\bibitem{NK}
K. Norrie, Actions and automorphisms of crossed modules. \emph{Bull. Soc. Math. France} {\bf 118} (1990), no.2, 129-146.

\bibitem{NK1}
K. Norrie, Crossed modules and analogues of group theorems. Ph.D. thesis, King's College. University of London, 1987.

\bibitem{PSV}
F. Pugliese, G. Sparano and L. Vitagliano, Multiplicative connections and their Lie theory. \emph{Commun. Contemp. Math.} {\bf 25} (2023), no.1, Paper No. 2150092, 36pp.

\bibitem{Ra1}
G. Rota, Baxter algebras and combinatorial identities. I.
\emph{Bull. Am. Math. Soc.} {\bf 75} (1969), 325-329.

\bibitem{Ra2}
G. Rota, Baxter algebras and combinatorial identities. II.
\emph{Bull. Am. Math. Soc.} {\bf 75} (1969), 330-334.

\bibitem{Ru}
W. Rump, Braces, radical rings, and the quantum Yang-Baxter equation. \emph{J. Algebra} {\bf 307} (2007), 153-170.

\bibitem{STS}
M. Semonov-Tian-Shansky, What is a classical $r$-matrix? \emph{Funct. Anal. Appl.} {\bf 17} (1983), 259-272.


\bibitem{Ta}
M. Takeuchi, Survey on matched pairs of groups. An elementary approach to the ESS-LYZ theory. \emph{Banach Center Publ.} {\bf 61} (2003), 305-331.



\bibitem{Y}
C.N. Yang, Some exact results for the many-body problem in one dimension
with repulsive delta-function interaction, \emph{Phys. Rev. Lett.} {\bf 19} (1967), 1312-1315.

\bibitem{WF}
F. Wagemann, Crossed modules. De Gruyter Stud. Math., 82. \emph{De Gruyter, Berlin,} (2021), xiv+393 pp.

\bibitem{Whi}
J.H.C. Whitehead, Note on a previous paper entitled "On adding relations to homotopy groups." \emph{Ann.of Math.(2)} {\bf 47} (1946), 806-810.

\bibitem{WHi}
J.H.C. Whitehead, Combinatorial homotopy. II. \emph{Bull. Amer. Math. Soc.} {\bf 55} (1949), 453-496.





\end{thebibliography}
 \end{document}